%% file: main.tex
\documentclass[11pt, letter]{article}
\usepackage{fullpage}
\newcommand{\ShowComment}{true}

\input{header.tex}

\title{Space Complexity of Vertex Connectivity Oracles\thanks{This work was supported by NSF grants CCF-1815316,  CCF-2221980, and CCF-2238138.  An 
extended abstract of this work appeared in the Proceedings of STOC 2022 as \emph{Optimal Vertex Connectivity Oracles}.}}
\date{}

\author{Seth Pettie\\ University of Michigan 
\and 
Thatchaphol Saranurak\\
University of Michigan
\and 
Longhui Yin\\
Tsinghua University}

\begin{document}
\maketitle
\pagenumbering{gobble}
\input{abstract.tex}
\pagebreak
\pagenumbering{arabic}

\input{intro.tex}

\input{body.tex}



\newcommand{\etalchar}[1]{$^{#1}$}


\end{document}

%% file: header.tex
\usepackage{amsthm}
\usepackage{graphicx} 

\usepackage{amsmath, amssymb, amsfonts, verbatim}
\usepackage{mathtools}
\usepackage{mathrsfs}
\usepackage{stmaryrd}
\usepackage[usenames,dvipsnames]{xcolor}
\usepackage{tcolorbox}

\usepackage{array}
\usepackage{multirow}
\usepackage[ruled, linesnumbered]{algorithm2e}

\usepackage{thmtools}
\definecolor{ForestGreen}{rgb}{0.1333,0.5451,0.1333}
\definecolor{DarkRed}{rgb}{0.8,0,0}
\definecolor{Red}{rgb}{1,0,0}
\usepackage[linktocpage=true,
pagebackref=true,colorlinks,
linkcolor=DarkRed,citecolor=ForestGreen,
bookmarks,bookmarksopen,bookmarksnumbered]{hyperref}

\usepackage[capitalise]{cleveref}

\declaretheorem[numberwithin=section]{theorem}
\declaretheorem[numberlike=theorem]{lemma}

\declaretheorem[numberlike=theorem]{corollary}
\declaretheorem[numberlike=theorem]{definition}

\declaretheorem[numberlike=theorem,style=remark]{remark}

\usepackage{url}

\newcommand{\rb}[2]{\raisebox{#1 mm}[0mm][0mm]{#2}}
\newcommand{\istrut}[2][0]{\rule[- #1 mm]{0mm}{#1 mm}\rule{0mm}{#2 mm}}

\newcommand{\ignore}[1]{}

\newcommand{\ceil}[1]{\left\lceil{#1}\right\rceil}
\newcommand{\floor}[1]{\left\lfloor{#1} \right\rfloor}

\newcommand{\poly}{\operatorname{poly}}

\newcommand{\E}{{\mathbb E\/}}

\newcommand{\lowerint}[1]{\lfloor#1\rfloor}

\global\long\def\poly{\mathrm{poly}}

\newcommand{\flow}{\operatorname{flow}}
\newcommand{\LG}{\operatorname{lg}}
\newcommand{\SM}{\operatorname{sm}}

\newcommand{\IN}{\operatorname{in}}
\newcommand{\OUT}{\operatorname{out}}

\newcommand{\econn}{\mathsf{e}\text{-}\mathsf{conn}}
\newcommand{\ecut}{\mathsf{e}\text{-}\mathsf{cut}}
\newcommand{\vconn}{\mathsf{v}\text{-}\mathsf{conn}}
\newcommand{\vcut}{\mathsf{v}\text{-}\mathsf{cut}}

\makeatletter

\newcommand{\Rmnum}[1]{\expandafter\@slowromancap\romannumeral #1@}
\makeatother

\renewcommand{\mod}{\ \mathrm{mod}\ }
\newcommand{\abs}[1]{\left|#1\right|} 

\newcommand{\algoF}{g}

\allowdisplaybreaks

\ifdefined\ShowComment
	\newcommand{\seth}[1]{\textcolor{blue}{SP: #1}}
	\newcommand{\longhui}[1]{\textcolor{red}{LY: #1}}
	\newcommand{\thatchaphol}[1]{\textcolor{purple}{TS: #1}}

\else
	\newcommand{\seth}[1]{}
	\newcommand{\thatchaphol}[1]{}
	\newcommand{\longhui}[1]{}
\fi

%% file: abstract.tex
\begin{abstract}
A \emph{$k$-vertex connectivity oracle} for an undirected graph $G$ 
is a data structure 
that, given $u,v\in V(G)$, reports $\min\{\kappa(u,v),k+1\}$,
where $\kappa(u,v)$ is the pairwise vertex connectivity 
between $u,v$.  There are three main measures of efficiency: construction time, query time, and space.
Prior work of Izsak and Nutov~\cite{IzsakN12}
produced a data structure of $O(kn\log n)$ words, 
which can even be encoded as a $O(k\log^3 n)$-bit labeling scheme,
that can answer vertex-connectivity 
queries in $O(k\log n)$ time.  
The construction time is polynomial, but unspecified.

In this paper we address the top three complexity measures.

\begin{description}
\item[Space.]
We prove that any $k$-vertex connectivity oracle requires $\Omega(kn)$ bits of space, 
for any $k\leq n$.
This proves that the Iszak-Nutov data structure 
is optimal up to polylogarithmic factors for every $k$, 
and that Nagamochi-Ibaraki~\cite{NagamochiI92} sparsifiers are optimal compression schemes
for $k$-vertex connectivity, up to a logarithmic factor.
In particular, whereas all edge-connectivities can be efficiently compressed (as a weighted 
$(n-1)$-edge Gomory-Hu tree), vertex connectivity admits no asymptotic compression: $\Omega(n^2)$ bits are necessary.
We design a variation on Izsak and Nutov's data structure that uses 
$O(\min\{kn\log n, m\log n\log k\})$ words of space.

\item[Query Time.]
We answer queries in $O(\log n)$ time,
improving on the $\Omega(k\log n)$ time bound of 
Izsak and Nutov~\cite{IzsakN12}. 
The main idea is to build instances of $\textsf{SetIntersection}$ data structures, 
with additional structure based on affine planes.
This structure allows for query time that is 
linear in the output size, which evades some
conditional lower bounds that are polynomial in the query set sizes~\cite{HenzingerKNS15,KopelowitzPP16}.

\item[Construction Time.]
Our data structure can be constructed in 
the time of $O(k^3\operatorname{polylog}(n))$ max-flow computations, namely $k^3m^{1+o(1)}$ time using the recent near-linear time flow algorithm of~\cite{ChenKLPGS22}.
The main technical contribution here is a fast algorithm to compute a $(1+\epsilon)$-approximate Gomory-Hu tree for \emph{element connectivity},
in the time of $O(\epsilon^{-1}\operatorname{polylog}(n))$ max-flow computations.  
Element connectivity is a notion that generalizes edge and vertex connectivity.
\end{description}
\end{abstract}

%% file: intro.tex
\section{Introduction}

Most measures of \emph{graph connectivity} can be 
computed in polynomial time, and much of the recent 
work in graph \emph{algorithms} aims at reducing
these complexities to some natural barrier, 
either near-linear~\cite{Karger00,KargerP09,ForsterNYSY20,GawrychowskiMW20,GawrychowskiMW21,MukhopadhyayN20,Saranurak21,BrandLLSS0W21},
max-flow time~\cite{LiP20,li2021vertex}, 
or matrix multiplication time~\cite{CheungLL13,AbboudGIKPTUW19}.
The recent $m^{1+o(1)}$ flow algorithm of~\cite{ChenKLPGS22} collapses the first two
classes, but it is still useful to distinguish between \emph{direct} algorithms 
and reductions to max-flow.

Graph connectivity has also been extensively studied from a \emph{structural} perspective, where the aim is to understand the structure of some or all minimum cuts.  
This genre includes
Gomory-Hu trees~\cite{GomoryH61}, the cactus representation~\cite{DinicKL76},
block-trees, SPQR-trees, and extensions to higher vertex connectivity~\cite{DiBattistaT96,KanevskyTBC91,cohen1993reinventing,PettieY21}, 
and many others~\cite{Benczur95b,BenczurG08,gusfield1993extracting,DinitzV94,DinitzV95,DinitzV00,DinitzN95,DinitzN99a,DinitzN99b,PicardQ80,GeorgiadisILP16,GeorgiadisILP18,FirmaniGILS16}.

In this paper we study the \emph{data structural} approach to understanding graph connectivity, which incorporates elements of the \emph{algorithmic} 
and \emph{structural} camps, but goes further in that we want to be able to efficiently \emph{query} the connectivity, e.g., either ask for the size of a min-cut or the min-cut itself.   

\medskip 

Suppose that we are given an undirected 
graph $G$ and wish to be able to accurately answer pairwise 
edge- and vertex-connectivities up to some threshold $k$.
Define $\lambda(u,v)$ and $\kappa(u,v)$
to be the maximum number of edge-disjoint and internally 
vertex-disjoint paths, resp., between $u$ and $v$. 
By Menger's theorem (or max-flow min-cut duality)
$\lambda(u,v)$ and $\kappa(u,v)$ are the minimum size of an \emph{edge cut} 
$C_0 \subseteq E$ and a \emph{mixed cut} $C_1 \subseteq E\cup (V-\{u,v\})$, 
respectively, that disconnect $u,v$ in $G-C_0$ and $G-C_1$.\footnote{If $\{u,v\}\not\in E(G)$, 
then there is always a cut $C_1$ consisting solely 
of vertices.  When $\{u,v\}\in E(G)$, there is always 
a cut $C_1$ consisting of $\{u,v\}$ and $\kappa(u,v)-1$ vertices.}
\begin{description}
\item[$\econn(u,v)$ : ] Return $\min\{\lambda(u,v),k+1\}$.
\item[$\ecut(u,v)$ : ] If $\lambda(u,v)\leq k$, 
return an edge-cut separating $u,v$ with size $\lambda(u,v)$.
\item[$\vconn(u,v)$ : ] Return $\min\{\kappa(u,v),k+1\}$.
\item[$\vcut(u,v)$ : ] If $\kappa(u,v)\leq k$, return a mixed cut
separating $u,v$ with size $\kappa(u,v)$.
\end{description}
The edge-connectivity problems are essentially solved.
A \emph{Gomory-Hu tree} $T$~\cite{GomoryH61} (aka \emph{cut equivalent tree})
is an edge-weighted tree such that
the bottleneck edge $e$ between any $u,v$ has weight $\lambda(u,v)$,
and the partition defined by $T-\{e\}$ corresponds to 
a $\lambda(u,v)$-size edge cut, which can be explicitly associated 
with $e$ if we wish to also answer $\ecut$ queries.
Bottleneck queries can be answered in $O(1)$ time
with $O(n)$ preprocessing.\footnote{Observe that since $G$ is unweighted, the edge weights
of $T$ (possible edge-connectivities) are integers in $\{0,1,\ldots,n-1\}$, 
which can be sorted in $O(n)$ time.  Once the
edge-weights are sorted, the bottleneck edge query problem is reducible in $O(n)$ time 
to the range-min query problem in an array; see~\cite[\S 3]{DemaineLW09}.  Range-min queries, in turn, 
can be answered in $O(1)$ time after $O(n)$ preprocessing; see Bender and Farach-Colton~\cite{BF-C00}.
This tradeoff cannot be attained in the comparison model, where $O(n)$ preprocessing
implies $\Omega(\alpha(n))$ query time~\cite{Pettie06}; cf.~\cite{DemaineLW09,Chaz87}.}
The time to compute the full Gomory-Hu tree in unweighted graphs is $m^{1+o(1)}$~\cite{AbboudK0PST22}, 
and when $k$ is small, there are faster algorithms~\cite{HariharanKP07} 
running in $\tilde{O}(m+nk^3)$-time.

Thus, $\econn$-queries can be answered 
in $O(1)$ time by a data 
structure occupying space $O(n)$.
It is a long-standing open problem
whether a similar result is possible for $\vconn$-queries.
In 1979 Schnorr~\cite{Schnorr79} proposed a cut-equivalent tree 
for \emph{roundtrip} flow ($\min\{f(u,v),f(v,u)\}$) 
in directed graphs, and in 
1990 Gusfield and Naor~\cite{GusfieldN90}
used Schnorr's result to build a Gomory-Hu-type tree
for minimum \emph{capacitated} vertex ``separations.''\footnote{Formally, 
they defined a $u$-$v$ \emph{separation} to be either $\{u\}$ or $\{v\}$ or some $u$-$v$ 
vertex cut $S\subseteq V-\{u,v\}$.  Note that this is not a meaningful concept
in unit vertex-capacity graphs as the minimum capacity separation always has value 1, irrespective of $G$.}
Benczur~\cite{Benczur95a} found errors in Schnorr~\cite{Schnorr79}
and Gusfield and Naor~\cite{GusfieldN90}, and proved more generally
that there is no \emph{cut}-equivalent tree for 
round-trip flow or vertex separations.
Benczur~\cite[pp. 505-506]{Benczur95a} suggested 
a way to get a \emph{flow}-equivalent tree for 
capacitated vertex connectivity using Cheng and Hu~\cite{ChengH91}, 
and the standard reduction from 
vertex connectivity to directed flow. 
This would yield a Gomory-Hu-type tree suitable
for answering $\vconn$ (but not $\vcut$) queries.  
This claim also turned out to be incorrect.
Hassin and Levin~\cite{hassin2007flow} gave an example of
a vertex-capacitated graph (with integer capacities in $[1,n^{O(1)}]$)
that has $\Omega(n^2)$ distinct pairwise vertex connectivities, which cannot
be captured by a Gomory-Hu-type tree representation.

When the underlying graph has \emph{unit capacity}, the counterexamples 
of~\cite{Benczur95a,hassin2007flow} do not rule out a representation
of vertex-connectivity using, say, 
$\tilde{O}(1)$ trees,
nor do they rule out some completely different $\tilde{O}(n)$-space structure
for answering $\vconn$-queries, independent of $k$.

\input{table}

Most prior data structures
supporting $\vconn$
queries are actually \emph{labeling schemes}.
A vertex 
\emph{labeling} $\ell:V\rightarrow \{0,1\}^*$ 
is created such that
a $\vconn(u,v)$ query is answered 
by inspecting only $\ell(u),\ell(v)$.  
Katz, Katz, Korman, and Peleg~\cite{katz2004labeling} 
initiated this line of research into labeling for connectivity.
They proved that the maximum label length
to answer $\vconn$ queries is 
$\Omega(k\log (n/k^3))$ and $O(2^k\log n)$.
To be specific, they give a class of graphs for which
a $\Theta(1/k^2)$-fraction of the vertices 
require $\Omega(k\log(n/k^3))$-bit labels. 
However, this does not imply any non-trivial 
bound on the average/total label length.
Their upper bound was subsequently improved 
to $O(k^3\log n)$~\cite{Korman10}
and then to $O(k^2\log n)$~\cite{hsu2009optimal}.\footnote{The labeling schemes of Katz et al.~\cite{katz2004labeling}, Korman~\cite{Korman10}, and Hsu and Lu~\cite{hsu2009optimal} differentiate between $\kappa(u,v)\geq k_0$
and $\kappa(u,v) < k_0$ using labels of size $O(2^{k_0}\log n)$, $O(k_0^2\log n)$,
and $O(k_0\log n)$, respectively.  They can be used to answer $\vconn$ queries 
by concatenating labels for all $k_0 = 1,2,\ldots,k$, thereby introducing an $O(k)$-factor overhead in~\cite{Korman10,hsu2009optimal}.}
Using a different approach, Izsak and Nutov~\cite{IzsakN12}
gave an $O(k\log^3 n)$-bit labeling scheme for $\vconn$ queries.  A centralized version of the data structure
takes $O(kn\log n)$ space,
and can be augmented to support 
$\vcut$ queries in $O(k^2n\log n)$ space.
(Following convention, the space of centralized data structures is measured in $O(\log n)$-bit \emph{words} 
unless specified otherwise, whereas labeling schemes are always measured in \emph{bits}.)

The labeling schemes~\cite{katz2004labeling,Korman10,hsu2009optimal,IzsakN12} 
focus on label-length, and do not discuss the construction time in detail, which is some large polynomial.

\subsection{New Results}

We resolve the long-standing question concerning the space complexity of representing vertex connectivity; 
see Hu~\cite{Hu69} and~\cite{Benczur95a,hassin2007flow}.  
We prove that any data structure answering 
$\vconn$ queries requires $\Omega(kn)$ bits, 
    and that the lower bound extends to threshold queries (decide whether or not $\kappa(u,v)\leq k$), equality queries (decide whether or not $\kappa(u,v)=k$), and approximate queries (distinguish $\kappa(u,v)<k-\Theta(\sqrt{k})$ from $\kappa(u,v)>k+\Theta(\sqrt{k})$).
    Our result has several notable consequences.
\begin{itemize}
\item Izsak and Nutov's~\cite{IzsakN12} centralized data structure is \emph{optimal} to 
    within a $\log^2 n$-factor, and even the \emph{average} 
    length of their labeling scheme cannot be improved by more than a $\log^3 n$-factor; cf.~\cite{katz2004labeling}.

\item The Nagamochi-Ibaraki~\cite{NagamochiI92}
sparsifier is a graph with $O(kn)$ edges, represented with $O(kn\log n)$ bits, 
that captures all vertex connectivities up to $k$.  Therefore,  
its space cannot be improved by more than a $\log n$-factor, 
even if the representation is not constrained to be a graph.

\item Whereas edge-connectivity admits compression to $O(n\log n)$ bits (as a Gomory-Hu tree),
vertex connectivity admits no compression by more than a constant factor; cf.~\cite{Benczur95a,hassin2007flow}.  
In particular, there exist dense graphs with $\Theta(n^2)$ edges such that storing all 
$\{\kappa(u,v)\}_{u,v}$ values requires $\Omega(n^2)$ bits.
\end{itemize}

We design a refined version of the Izsak-Nutov data structure that improves the space bound
on sparse graphs, and the query time and construction time in general.  
Our data structure occupies
$O(\min\{kn\log n, m\log n\log k\})$ space, 
i.e., it is an improvement
when $m < kn/\log k$.

Our data structure 
answers queries in 
$O(\log n)$ time, improving $O(k\log n)$~\cite{IzsakN12}.
Each vertex $v$ is associated with a set $S_v$ with $|S_v|=\Theta(k\log n)$
and the query time for $\kappa(u,v)$ is dominated by computing a \textsf{SetIntersection} query 
$S_u\cap S_v$, which reports every element of the intersection.
Some conditional lower bounds~\cite{KopelowitzPP16,HenzingerKNS15} 
show that space- and preprocessing-efficient data structures for \textsf{SetIntersection} must have 
$k^{1-o(1)}$ query time.  
We give a method for generating \emph{structured} sets $\{S_v\}_{v\in V}$ that allow us to answer
queries in optimal $O(|S_u\cap S_v|)=O(\log n)$ time, while still being sufficiently \emph{random}
to allow the correctness analysis of~\cite{IzsakN12} to go through.

The construction time for our data structure 
is $O(k^3\operatorname{polylog}(n))\cdot T_{\flow}(m)$, 
where $T_{\flow}(m)$ is the time for one max-flow computation.
The main subproblem solved in the construction algorithm is computing a 
$(1+\epsilon)$-approximate Gomory-Hu tree for \emph{element connectivity}.

\medskip 

If one looks only at the tradeoff between space and query time for $\vconn$ queries, 
there are now four incomparable tradeoffs.
\begin{itemize}
    \item A lookup table, with $O(n^2)$ 
    space and $O(1)$ query time.
    \item A Nagamochi-Ibaraki~\cite{NagamochiI92} sparsified graph, 
    with $O(kn)$ space and 
    \emph{max-flow} (up to flow value $k$) 
    query time, namely $O(\min\{k^2n, (kn)^{1+o(1)}\})$ time~\cite{ChenKLPGS22}.
    \item The Izsak-Nutov~\cite{IzsakN12} structure, as improved in this paper, 
    with $O(kn\log n)$ space and $O(\log n)$ query time.
    \item Nutov's structure~\cite{Nutov22}, with  $O(k^2 n)$ space and $O(1)$ query time.
\end{itemize}

\subsection{Organization}

Section~\ref{sect:preliminaries}
covers some preliminary concepts
such as vertex connectivity,
element connectivity, Gomory-Hu
trees, and the (vertex-cut version of) the isolating cuts
lemma of~\cite{LiP20,li2021vertex}.

Section~\ref{sec:lower_bound}
presents the lower bound on representations of vertex connectivity.
Section~\ref{sect:vconn-oracle}
presents a space- and 
query-efficient 
vertex connectivity oracle
based on Izsak and Nutov's~\cite{IzsakN12} labeling scheme.  
In order to build this oracle efficiently, we need to compute 
Gomory-Hu trees that capture all
element connectivities up to $k$. 
In Section~\ref{sec:GH_tree_construction} we 
show how to build one such Gomory-Hu tree in 
$O(k\operatorname{polylog}(n))\cdot T_{\flow}(m)$ time.

Section~\ref{sect:conclusion} concludes with 
some open problems.

%% file: table.tex
\providecommand{\tabularnewline}{\\}

\begin{table}

\begin{tabular}{|p{.25\textwidth}|l|l|l|}
\hline 
\textsf{Citation} & \textsf{Space (Words)} \hfill \ \ \textsf{Labeling (Bits)}& \textsf{Query} & \textsf{Construction} \tabularnewline\hline\hline 
Block-tree \hfill $k=1$  & $O(n)$ & $O(1)$\istrut[2]{4} & $O(m)$ \tabularnewline\hline
SPQR~\cite{DiBattistaT96} \hfill $k=2$  & $O(n)$ & $O(1)$\istrut[2]{4} & $O(m)$ \tabularnewline
\hline 
\cite{KanevskyTBC91} \hfill $k=3$ & $O(n)$ & $O(1)$\istrut[2]{4} & $O(m+n\alpha(n))$ \tabularnewline
\hline 
\rb{-2.5}{\cite{katz2004labeling}} & $O(n2^{k})$ \hfill $O(2^k\log n)$ & $O(2^k)$\istrut[2]{4} & $\poly(n)$ \tabularnewline\cline{2-4}
        & \hfill $\Omega(k\log(n/k^3))$ (w.c.) & \emph{any}\istrut[2]{4} & \emph{any} \tabularnewline\hline 
\cite{Korman10} & $O(nk^{3})$ \hfill $O(k^3\log n)$ & $O(k\log k)$\istrut[2]{4} & $\poly(n)$ \tabularnewline
\hline 
\cite{hsu2009optimal} & $O(nk^2)$ \hfill $O(k^2\log n)$\istrut[2]{4} & $O(\log k)$ & $\poly(n)$ \tabularnewline\hline 
\cite{IzsakN12} & $O(nk\log n)$ \hfill $O(k\log^3 n)$ & 
$O(k\log n)$\istrut[2]{4} & $\poly(n)$ \tabularnewline\hline 
\cite{AbboudGIKPTUW19} & $O(n^{2})$ \hfill $O(n\log n)$ & $O(1)$\istrut[2]{4} & $O(nk)^{\omega}$ \tabularnewline\hline 
\cite{PettieY21,Nutov22} \hfill $k=\kappa(G)$ & $O(n\kappa(G))$ \hfill $O(\kappa(G)\log n)$ & $O(1)$\istrut[2]{4} & $\tilde{O}(m+n\poly(\kappa))$ \tabularnewline\hline
\cite{Nutov22} \hfill & $O(nk^2)$ & $O(1)$\istrut[2]{4} & $\poly(n)$ 
\tabularnewline\hline 
\rb{-5}{\textbf{New}}  & $O(\min\{kn\log n,$ \ \    \hfill \rb{-2.5}{$O(k\log^3 n)$} & \rb{-2.5}{$O(\log n)$}\istrut{4} & \rb{-2.5}{$O(k^3\log^7n)\cdot T_{\flow}(m)$} \\
        & \ \ \ \ $m\log n\log k\})$ &\istrut[2]{0}&\\\cline{2-4}
            & $\Omega(nk/\log n)$ \hfill $\Omega(k)$ (avg.) & \emph{any}\istrut[2]{4} & \emph{any}
\tabularnewline\hline\hline
\end{tabular}
\caption{\small A history of vertex-connectivity oracles.  By convention, space for centralized data structures is measured in $O(\log n)$-bit words and space
for labeling schemes is measured in bits.
With a couple exceptions, all data structures and lower bounds are for
$\vconn(u,v) = \min\{\kappa(u,v),k+1\}$ queries, where $k$ is an arbitrary parameter.
The constructions based on block trees, SPQR trees~\cite{DiBattistaT96}, and~\cite{KanevskyTBC91} only work for $k\in\{1,2,3\}$, and the structures of
\cite{PettieY21,Nutov22} 
only work when $k=\kappa(G)$ is
the global minimum vertex connectivity of $G$.  
The $\Omega(k\log(n/k^3))$ lower bound of Katz et al.~\cite{katz2004labeling}
is for the worst-case (longest) label length; it implies nothing on average length.
The new $\Omega(nk)$-bit lower bound is for the total size of the data 
structure, and implies an $\Omega(k)$-bit lower bound on the average length 
of any labeling scheme.}\label{table:history}

\end{table}

%% file: body.tex
\section{Preliminaries}\label{sect:preliminaries}

\subsection{Vertex Connectivity, Element Connectivity, and Gomory-Hu Trees}

Define $\kappa_G(u,v)$ to be the vertex connectivity of $u,v$ in $G=(V,E)$,
i.e., the maximum number of internally vertex-disjoint paths 
between $u$ and $v$.  By Menger's theorem~\cite{Menger27}, 
this is the size of the smallest mixed cut 
$C\subset (E\cup (V-\{u,v\}))$ whose
removal disconnects $u,v$.\footnote{Sometimes only ``pure'' vertex cuts $C\subset V-\{u,v\}$ are considered and $\kappa_G(u,v)$ is left undefined
or artificially defined to be $|V|-1$ whenever $u,v$ are adjacent in $G$. This definition fails to distinguish highly connected and poorly connected pairs of vertices that happen to be adjacent. It is possible to splice a vertex into every edge so that all \emph{mixed} cuts 
become \emph{pure} vertex cuts, 
but this increases the number of 
vertices in the graph.}
\emph{Element connectivity} is a useful generalization of vertex- and edge-connectivity.
The name was coined by Jain, Mandoiu, Vazirani, and Williamson~\cite{JainMVW02}, but the concept was latent in earlier work~\cite{FrankI93, EvenI98}. 
See~\cite{ChekuriK14,chekuri2015some,chekuri2015element,ChuzhoyK12} for recent work.
Let $U\subseteq V$ be a set of terminals. 
When $u,v\in U$, 
define $\kappa'_{G,U}(u,v)$ to be 
the maximum number of $(E\cup (V-U))$-disjoint paths between $u$ and $v$,
i.e., the paths may intersect, but only at terminals.
By duality, this is equivalently the size of the smallest 
mixed cut $C\subset (E\cup (V-U))$ whose removal disconnects $u,v$.
Observe that when $U=V$, 
$\kappa'_{G,U}(u,v)=\lambda_G(u,v)$ is 
the same as edge-connectivity, and 
when $U=\{u,v\}$, $\kappa'_{G,U}(u,v)=\kappa_G(u,v)$ 
is exactly the vertex connectivity of $u,v$.
More generally, we have the equality $\kappa_{G,U}'(u,v)=\kappa_G(u,v)$ 
whenever $U$ {\bf \emph{captures}} $u,v$, as defined below.

\begin{definition}\label{def:capture}
    A vertex set $U\subseteq V$ \textbf{captures} $u,v$ if $u,v\in U$ and 
    there exists a mixed $u$-$v$ cut $C\subset E\cup (V-\{u,v\})$ with 
    $|C|=\kappa_{G}(u,v)$ and $C\cap U=\emptyset$.
\end{definition}

\cref{lem:relation_between_connectivities},
follows directly from \cref{def:capture} and the definitions of $\kappa_G, \kappa_{G,U}'$.

\begin{lemma}\label{lem:relation_between_connectivities}
For any $U\supset \{u,v\}$, $\kappa'_{G,U}(u, v) \geq \kappa_{G}(u, v)$, 
with equality holding if and only if $U$ captures $u,v$.
\end{lemma}

If $C$ is a mixed cut, the connected components of $G-C$ are called \emph{sides} of $C$.  
Note that minimum edge-cuts have two sides but minimum vertex-cuts may have an unbounded number of sides.

We extend the definition of element connectivity to 
vertex sets $A,B\subset U$ as follows. 
Define $\kappa'_{G, U}(A, B)$ to be the size of 
the smallest mixed $A$-$B$ cut 
$C\subset (E\cup (V - U))$,
i.e., $G-C$ disconnects $a$ and $b$, 
for every $(a,b)\in A\times B$. 
Note that it may be the case that $A$ and $B$ are not contained in a single side of $C$.

For any $A\subset V$, define $\partial A$ to be the set of vertices adjacent to $A$ in $G$, 
but not in $A$.  Thus, whenever $A\cup\partial A\neq V$, 
$\partial A$ is a vertex cut disconnecting $A$ from 
$V-(A\cup \partial A)$.
It also satisfies a well-known submodularity property, 
which will be used in the proof of \cref{lem:isolating_cut_lemma_with_forbidden_terminals} and \cref{sec:GH_tree_construction}.

\begin{lemma}
\label{lem:submodularity-of-partial}
For any two subsets $A$, $B$ of $V$, $\abs{\partial A} + \abs{\partial B} \geq \abs{\partial (A\cup B)} + \abs{\partial (A\cap B)}$.
\end{lemma}

\ignore{
\begin{proof}
Fixing $v\in\partial A$, we have three cases.
When $v\in \partial A \cap \partial B$, $v$ contributes $2$ to the left hand side 
and at most $2$ to the right hand side. 
When $v\in \partial A \cap B$, we have $v\notin \partial B$ and $v\notin \partial(A\cup B)$,
so it contributes 1 to the left hand side and at most 1 to the right.
For $v\in \partial A - (B\cup \partial B)$, $v$ contributes $1$ in the left hand side,
and at most 1 to the right, since $v\notin \partial(A\cap B)$.
The same analysis applies to any $v\in \partial B$, and when $v\not\in \partial A\cup \partial B$, $v$ contributes zero to both sides.
\end{proof}
}

When the terminal set $U$ is an independent set in $G$, 
we can focus on element cuts $C\subseteq V-U$ that consist solely of vertices.
This is because any edge that might appear in $C$ can be substituted with
a non-terminal endpoint, which exists since $U$ is independent.
This observation simplifies some aspects of \cref{lem:isolating_cut_lemma_with_forbidden_terminals} and the algorithms of \cref{sec:GH_tree_construction}.
\begin{lemma}
\label{lem:partial-of-side-and-its-cut}
Suppose $U$ is an independent terminal set in $G(V, E)$, $A,B\subset U$, and 
$C\subset (V - U)$ is an $A$-$B$ cut. 
Define $D_A,D_B$ to be the union of sides of $C$ containing an $A$-terminal and $B$-terminal, respectively, so $D_A\cap D_B=\emptyset$.
Let $D$ be any side of $C$.

Then $\partial D \subseteq C$. 
Moreover, if $C$ is a \emph{minimum} $A$-$B$ cut, then $\partial D_A = \partial D_B = C$.    
\end{lemma}

\begin{proof}
Since $D$ is a connected component in $G-C$ and every $u\in \partial D$ is connected to $D$, 
it must be that $u\in C$.  Therefore $\partial D_A\subseteq C$ and $\partial D_B\subseteq C$. 
Observe that $\partial D_A$ is one $A$-$B$ cut, so by the minimality of $C$, 
it must be that $\partial D_A = C$.
\end{proof}

It is well known that Gomory-Hu trees exist for element connectivity; see~\cite{chekuri2015element, schrijver-book}. 
We use the following
general definition:
\begin{definition}
\label{def:elem_conn_gh_tree}
A $k$-Gomory-Hu tree for element connectivity w.r.t. graph 
$G$ and terminal set $U$ is a triple $(T,f,C)$ satisfying
\begin{itemize}
    \item (Flow equivalency) $T=(V_T,E_T,w:E_T\rightarrow [1,k])$ 
    is a weighted tree with embedding function 
    $f:U\rightarrow V_T$. 
    If $f(u)=f(v)$ then $\kappa'_{G,U}(u,v) > k$; 
    otherwise,
    $\kappa'_{G,U}(u,v) = \min_{e\in T(f(u),f(v))} w(e)$,
    where $T(x,y)$ is the unique $T$-path between $x$ and $y$.

    \item (Cut equivalency) $C:E_T\rightarrow 2^{E\cup (V-U)}$.
    For any edge $e\in T(f(u),f(v))$, 
    $C(e)$ is a cut of size $|C(e)|=w(e)$
    disconnecting $u,v$ in $G$.  
    (Each of the two connected components
    in $T-e$ corresponds to the union of some subset of the sides of $C(e)$
    under the projection $f$.)
\end{itemize}
\end{definition}

Note that $f$ is unnecessary when $k=\infty$ (in this case $V_T=U$),
and that $C$ is unnecessary if we are only interested in answering
$\min\{\kappa'_{G,U}(u,v),k+1\}$-queries.  This definition can be extended
to $(1+\epsilon)$-approximations by relaxing 
flow equivalency to
\[
\kappa'_{G, U}(u, v)\leq \min_{e\in T(f(u),f(v))} w(e) \leq (1+\epsilon)\kappa'_{G,U}(u,v).
\]

\subsection{Isolating Cut Lemma and Max-Flows}

Li and Panigrahi's~\cite{LiP20} \emph{isolating cuts} lemma finds, 
for a terminal set $I$, the minimum edge-cut 
separating $v\in I$ from $I-\{v\}$
in time proportional to $O(\log |I|)$ max-flows.  
It was generalized to vertex
connectivity in~\cite{li2021vertex} when $I$ is an independent set and 
here we generalize this lemma slightly
further to accommodate a \emph{forbidden} vertex set, which arises in 
the algorithms of \cref{sec:GH_tree_construction}.
In \cref{lem:isolating_cut_lemma_with_forbidden_terminals} 
we are searching for a set of cuts that are isolating w.r.t.~$I$, but avoid
a larger set $I\cup F$, where $F$ is the forbidden set.  

\begin{lemma}[Isolating cuts, vertex version with forbidden terminals]
\label{lem:isolating_cut_lemma_with_forbidden_terminals}
Let $I\cup F$ be an independent set in graph $G=(V, E)$, and $I\cap F=\emptyset$. 
We want to find, for each $v\in I$, a set $S_v\subset V$ such that 
(i) $S_v\cap I=\{v\}$ and $(\partial S_v)\cap (I\cup F)=\emptyset$,
(ii) it minimizes $\abs{\partial S_v}$, and
(iii) subject to criterion (ii), it minimizes $|S_v|$.
Then $\{S_v\}_{v\in I}$ can be computed 
with $O(\log|I|)$ calls to max-flow on graphs with 
$O(m)$ edges, $O(m)$ vertices, 
and capacities in $\{1,\infty\}$.
Moreover, each $S_v$ is unique, connected, and disjoint from every other $S_u$, $u\neq v$.
\end{lemma}

\begin{proof}
Our proof follows \cite[Lemma 4.2]{li2021vertex}, 
modified to accommodate the forbidden set $F$.
The set family $\{\{v\}\}_{v\in I}$ satisfies (i), 
so such sets exist. 
Among all set families satisfying (i), 
suppose $\{S_v\}_{v \in I}$ is our desired output, which 
minimizes $\abs{\partial S_v}$ for every $v\in I$, and then minimizes $\abs{S_v}$.
We first argue uniqueness.  Suppose, for some $v\in I$ there are $S_1,S_2$
that satisfy (i) and (ii), but have the same (minimum) size $|S_1|=|S_2|$.
It follows that $S_1\cap S_2$ and $S_1\cup S_2$ also satisfy (i),
and by (ii) it must be that $|\partial (S_1\cup S_2)| \geq |\partial S_1| = |\partial S_2|$.
By \cref{lem:submodularity-of-partial}, we have
\[\abs{\partial (S_1\cap S_2)} \leq \abs{\partial S_1} + \abs{\partial S_2} - \abs{\partial (S_1\cup S_2)} \leq \abs{\partial S_1} = \abs{\partial S_2},
\]
a contradiction since $S_1\cap S_2$ satisfies (i), and 
either $|\partial(S_1\cap S_2)|<|\partial S_1|=|\partial S_2|$ or $|S_1\cap S_2|<|S_1|=|S_2|$.
The graph induced by $S_v$ is 
clearly connected, as components not 
containing $v$ may be discarded.

Assign the vertices of $I$ distinct $\ceil{\log|I|}$-bit labels. 
For $j\in [\ceil{\log|I|}]$, let $(A_j,B_j)$ be the partition of $I$ according to the $j$th bit of their labels.
We compute the minimum size $A_j$-$B_j$ cut $C_j\subset V-(I\cup F)$ using 
the standard reduction from vertex/element connectivity to max-flow.\footnote{Specifically, replace each $v\in V$ by two vertices $v_{\IN},v_{\OUT}$ and an edge $(v_{\IN},v_{\OUT})$, and replace each undirected edge $\{u,v\}$ with two directed edges $(u_{\OUT},v_{\IN}),(v_{\OUT},u_{\IN})$.
Introduce source $s$ and sink $t$, with edges 
$(s,a_{\IN})$ and 
$(b_{\OUT},t)$ for every $a\in A_j,b\in B_j$.  All edges have capacity $\infty$, 
except for $(v_{\IN},v_{\OUT})$ with $v\in V-(I\cup F)$, which have unit capacity.
Clearly every $s$-$t$ cut with finite value 
corresponds to an $A$-$B$ vertex cut consisting of eligible vertices in $V-(I\cup F)$.}
Now define $U_v$ to be the connected component containing $v$ in 
$G - \bigcup_{j=1}^{\ceil{\log\abs{I}}} C_j$. 
Observe that $\{U_v\}_{v\in I}$ satisfies (i), namely $U_v \cap I = \{v\}$, since
for every $u\in I-\{v\}$ there is some bit-position $j$ where $u,v$ differ, 
implying $C_j$ is a $u$-$v$ cut.  

We claim $S_v \subseteq U_v$ for all $v\in I$.
Define $T_{v,j}$ to be the union of the sides of $C_j$ 
containing $A_j$-terminals, if $v\in A_j$, or $B_j$-terminals, if $v\in B_j$.
Note that $\partial S_v$ is the \emph{minimum} $v$-$(I-\{v\})$ cut
while $\partial (S_v \cap T_{v,j})$ is \emph{some} $v$-$(I-\{v\})$ cut.
By \cref{lem:partial-of-side-and-its-cut},
$\partial T_{v,j}=C_j$ is the \emph{minimum} $A_j$-$B_j$ cut
while $\partial(T_{v,j}\cup S_v)$ is \emph{some} $A_j$-$B_j$ cut.
Then
\begin{align*}
    \abs{\partial S_v} &\leq \abs{\partial (S_v \cap T_{v,j})},\\
    \abs{\partial T_{v,j}} &\leq \abs{\partial(T_{v,j}\cup S_v)},\\
    \abs{\partial S_v} + \abs{\partial T_{v, j}} &\geq \abs{\partial (S_v\cap T_{v, j})} + \abs{\partial(S_v\cap T_{v, j})}, & \text{\cref{lem:submodularity-of-partial}}
\end{align*}
which implies that all hold with equality.
Since $S_v\cap T_{v,j}$ also satisfies criteria 
(i) and (ii), criterion (iii) implies that
we must have $|S_v|=|S_v\cap T_{v,j}|$.
It follows that $S_v\subseteq T_{v,j}$,
and hence $S_v \subseteq \cap_{j=1}^{\ceil{\log|I|}} T_{v,j}$.  As $U_v$ is the connected component of $\cap_{j=1}^{\ceil{\log|I|}} T_{v,j}$ containing $v$,
and $S_v$ is connected, we have $S_v\subseteq U_v$.

We compute $S_v$ using the standard reduction from vertex/element connectivity to max-flow.  In particular,
we are looking for a vertex cut avoiding $I\cup F$ 
and separating $v$ and $\partial U_v$ in the graph induced 
by $U_v\cup \partial U_v$.  The graphs
induced by $\{U_v\cup \partial U_v\}$ are edge-disjoint,
so the combined cost of these calls to max-flow 
is equivalent to one max-flow computation on an 
$m$-edge graph.
\end{proof}

\section{Space Lower Bound on Vertex Connectivity Oracles}\label{sec:lower_bound}

\paragraph{Informal strategy and intuitions.}
We use an error correcting code-type argument to derive
an $\Omega(n^2)$-bit lower bound in the case that $k=n$.
By taking the union of $n/k$ disjoint copies of this construction on $\Theta(k)$-vertex graphs, we derive the $\Omega(kn)$-bit lower bound for general $k\in[1,n]$.
The idea is to show the existence of a \emph{codebook}
$\mathcal{T}$ of $n\times n$ 0--1 matrices with the following property.  
Each $T\in \mathcal{T}$ can be represented by a certain graph $G[T]$ on $O(n)$ vertices.
Assuming we have a $b$-bit vertex connectivity oracle for this graph, 
we query $\kappa_{G[T]}(u,v)$ for all pairs $u,v$. 
From these values we reconstruct a \emph{different}
0--1 matrix $\tilde{T}\neq T$, 
which is within the decoding radius
of $\mathcal{T}$, and therefore lets us deduce the identity of the original matrix $T$.
(The \emph{decoding radius} of $\mathcal{T}$ is half the minimum Hamming distance between $T,T'\in \mathcal{T}$.
The \emph{Hamming distance} between two matrices or two vectors is the number of positions 
where they differ.)
In other words, it must be that the vertex connectivity oracles of every $G[T],G[T']$, $T,T'\in \mathcal{T}$,
are distinct, and therefore $b\ge \log_2|\mathcal{T}|$.  

A key technical idea is to show that each $T\in \mathcal{T}$ in the codebook can, 
in a certain sense, be \emph{approximately} factored as the 
product of two rectangular 0--1 matrices, with addition and multiplication over $\mathbb{Z}$.
The specifics of the construction are detailed in \cref{thm:lb}.

\begin{theorem}\label{thm:lb}
There exists a constant $c$ and a subset 
$\mathcal{T} \subseteq \{0,1\}^{n\times n}$ of 0--1 
matrices having the following properties.
\begin{enumerate}
\item (Code Properties) $|\mathcal{T}|=2^{n^2/6}$, each row of each 
$T\in \mathcal{T}$ has Hamming weight
exactly $n/2$, and every two $T,T'\in\mathcal{T}$ has Hamming distance
greater than $n^2/6$, i.e., 
the decoding radius is at most $n^2/12$.
\item (Matrix Decomposition) For every $T\in\mathcal{T}$, there exists 
$A\in \{0,1\}^{n\times cn}, B\in\{0,1\}^{cn\times n}$ such that $C=AB$ approximately encodes $T$ in the following sense.  Let $\hat{T}\in\{0,1\}^{n\times n}$ be such that 
\[
\hat{T}(i,j) = \left\{
\begin{array}{ll}
0 & \mbox{if $C(i,j) < 2n$}\\
1 & \mbox{if $C(i,j) \geq 2n$}
\end{array}
\right.\]
Then $T,\hat{T}$ have Hamming distance less than $n^2/12$. Moreover,  $C(i,j) \le 2.1n$ for all $i,j$.
\item (Vertex Connectivity) For every $T\in\mathcal{T}$, there is an undirected graph 
$G[T] = (X\cup Y\cup Z, E)$ 
where $X=\{x_i\}_{i\in [n]}, Y=\{y_i\}_{i\in[cn]},Z=\{z_i\}_{i\in[n]}$
and $E\subset (X\times Y) \cup (Y\times Z)$,
that approximately 
encodes $T$ in the following sense.
Let $\tilde{T}\in\{0,1\}^{n\times n}$ be 
defined such that
\[
\tilde{T}(i,j) = \left\{
\begin{array}{ll}
0 & \mbox{if $\kappa(x_i,z_j) < 4n$}\\
1 & \mbox{if $\kappa(x_i,z_j) \geq 4n$.}
\end{array}
\right.
\]
Then the Hamming distance between 
$T,\tilde{T}$ is less than $n^2/12$.
\end{enumerate}
\end{theorem}

We prove \cref{thm:lb} through a series of lemmas.

\begin{lemma}\label{lem:code-properties}
Fix a sufficiently large even $n$.
There exists a codebook $\mathcal{T}$
where $|\mathcal{T}| = 2^{n^2/6}$, 
each row of each $T\in\mathcal{T}$ has
Hamming weight $n/2$, and every two $T,T'\in \mathcal{T}$ have Hamming distance at 
least $n^2/6$.
\end{lemma} 

\begin{proof}
Pick a set $\mathcal{T}_0$ 
of $2\cdot 2^{n^2/6}$ 0--1 matrices 
uniformly at random.
Obtain $\mathcal{T}_1$ by discarding any $T\in\mathcal{T}_0$ if
\begin{enumerate}
\item[I.] $T$ is within Hamming distance $n^2/5$ of another $T'\in\mathcal{T}_0$, or
\item[II.] $T$ is \emph{not} within Hamming distance $n\sqrt{n\ln n}$ of a matrix
whose rows have Hamming weight $n/2$.
\end{enumerate}
We argue that conditions I and II each cause a $o(1)$-fraction of $\mathcal{T}_0$
to be discarded, so $|\mathcal{T}_1|\geq 2^{n^2/6}$.
The number of matrices within Hamming distance $\delta n^2$ of of $T$, $\delta<1/2$, is
\begin{align*}
    \sum_{t=1}^{\delta n^2} {n^2\choose t} 
    &= \poly(n) {n^2 \choose \delta n^2}\\
    &= \poly(n) \frac{(n^2/e)^{n^2}}{(\delta n^2/e)^{\delta n^2}((1-\delta)n^2/e)^{(1-\delta)n^2}} & \text{(Stirling's approximation)}\\
    &= \poly(n) 2^{n^2 H(\delta)},
\end{align*}
where $H(\delta)=\delta\log_2(\delta^{-1})+(1-\delta)\log_2((1-\delta)^{-1})$ is the 
binary entropy function.  
For $\delta=1/5$, $H(\delta)<0.73$, so 
the probability that $\mathcal{T}_0-\{T\}$ 
hits the Hamming ball
of radius $n^2/5$ around $T$ is, by a union bound, at most
\[
\poly(n) 2^{n^2(H(1/5)-1)}\cdot 
2\cdot 2^{n^2/6}
< 2^{-0.1 n^2}
= o(1).
\]
I.e., condition I causes a $o(1)$-fraction of the matrices to be discarded in expectation.
Turning to condition II, if $T\in\mathcal{T}_0$ is chosen uniformly at random, the probability that
its first row has Hamming weight 
outside the range $[n/2 - \sqrt{n\ln n}, n/2 + \sqrt{n\ln n}]$ 
is at most $2\cdot \exp(-2(\sqrt{n\ln n})^2/n) = 2/n^2$.\footnote{Here we use the simplified Chernoff-Hoeffding bound, 
$\Pr(X > \E(X)+t),\Pr(X<\E(X)-t) < \exp(-2t^2/n)$, 
where $X$ is the sum of $n$ independent random variables in $[0,1]$.}
If all of $T$'s rows fail to have this property then $T$ appears in $\mathcal{T}_1$.
Hence, at most a $2/n=o(1)$ fraction of the matrices are discarded due to condition II.
The codebook $\mathcal{T}$ is obtained by replacing each $T\in\mathcal{T}_1$
with its nearest matrix with Hamming weight $n/2$ in each row.
Conditions I and II imply that any two $T,T'\in\mathcal{T}$ have 
Hamming distance at least 
$n^2/5 - 2n\sqrt{n\ln n} > n^2/6$.
\end{proof}

\begin{lemma}\label{lem:AB-construction}
For every $T\in\mathcal{T}$, there exists 
$A\in \{0,1\}^{n\times cn}, B\in\{0,1\}^{cn\times n}$ such that $C=AB$ approximately encodes $T$ in the following way.  Let $\hat{T}\in\{0,1\}^{n\times n}$ be such that 
\begin{align*}
\hat{T}(i,j) &= \left\{
\begin{array}{ll}
0 & \mbox{if $C(i,j) < \tau$}\\
1 & \mbox{if $C(i,j) \geq \tau$,}
\end{array}
\right.\\
\text{for any threshold } \tau &= 2n\pm O(1).
\end{align*}
Then $T,\hat{T}$ have Hamming distance less than $n^2/12$. 
Moreover,  $C(i,j) \le 2.1n$ for all $i,j$.
\end{lemma}

\begin{proof}
The construction of $A,B$ is probabilistic.  
We pick $B$ uniformly at random such that each row has Hamming weight $n/2$,
then choose $A$ depending on $T$ and $B$.
We call the pair $(i,k)$ \emph{eligible} 
if the vector $B(k,\cdot)$ has an unusually high 
agreement with $T(i,\cdot)$, in particular,
\[
\text{$(i,k)$ eligible means } 
|\{j : B(k,j) = T(i,j)\}| \geq n/2 + \sqrt{n}.
\]
For each row $i$, we set $A(i,k)=1$ for the first $4n$ 
values of $k$ for which $(i,k)$ is eligible.
For the time being, assume that 
$\Pr((i,k) \text{ is eligible}) \geq p$ for some absolute constant $p$.  This is proved separately in Lemma~\ref{lem:eligible}.
If $c=4p^{-1}+1$ then we can build $A$ with high probability.
Fix an $i$ and let $X$ be the number of eligible pairs $(i,k)$.
Then $\E(X) = pcn = 4n + pn$ and row $A(i,\cdot)$ will fail
to have weight $4n$ when $X < \E(x) - pn$.  By the Chernoff-Hoeffding bound, this happens with probability $\exp(-\Omega(p^2 n))$.  Therefore, with probability $1-\exp(-\Omega(p^2 n))$ the matrix $A$ can be constructed, with all rows having Hamming weight $4n$.

Now let $C=AB$ and consider the random variable $C(i,j)=\sum_k A(i,k)\cdot B(k,j)$.  
Conditioning on $A(i,k)=1$ means that $(i,k)$ is eligible and hence $B(k,j)=T(i,j)$ for $n/2+\sqrt{n}$ values of $j$.  Thus we have
\[
\Pr(A(i,k)\cdot B(k,j) = T(i,j) \;|\; A(i,k)=1) \geq \frac{n/2 + \sqrt{n}}{n} = 1/2 + 1/\sqrt{n}.
\]
It follows that 
\begin{align*}
    \E(C(i,j)) 
    &= \sum_{k} \E(A(i,k)\cdot B(k,j)) \\
    &= \sum_{k : A(i,k)=1} \E(A(i,k)\cdot B(k,j) \mid A(i,k)=1)
    \;
    \left\{\begin{array}{ll}
    \geq 4n(1/2 + 1/\sqrt{n}) & \mbox{when $T(i,j)=1$,}\\
    \leq 4n(1/2 - 1/\sqrt{n}) & \mbox{when $T(i,j)=0$.}
    \end{array}
    \right.
\end{align*}
In other words, if $T(i,j)=1$ then $C(i,j)$ 
dominates $\text{Binom}(4n,1/2+1/\sqrt{n})$ 
and if $T(i,j)=0$ then $\text{Binom}(4n,1/2-1/\sqrt{n})$ dominates $C(i,j)$.
Since $\E(C(i,j)) = 2n + 4\sqrt{n}$ or $2n-4\sqrt{n}$,
the probability that $\hat{T}(i,j)\neq T(i,j)$ is 
misclassified is the probability that $C(i,j)$ deviates from its expectation by 
$(2n + 4\sqrt{n})-\tau$ or 
$\tau - (2n-4\sqrt{n})$, 
both of which are $4\sqrt{n}\pm O(1)$.
By Chernoff-Hoeffding, the probability this happens
is at most $\exp(2(4\sqrt{n}-O(1))^2/4n) = (1+o(1))\exp(-8)$.
Thus, in expectation $\hat{T}$ and $T$ have Hamming
distance much less than $n^2/12$, 
the decoding radius of $\mathcal{T}$.
Finally, by Chernoff-Hoeffding bounds again,
$\Pr(C(i,j) \geq 2.1n) = \exp(-\Omega(n))$.
This establishes the \emph{existence} of matrices 
$A,B$ having the stated properties.
\end{proof}

\begin{lemma}\label{lem:eligible}
    $\Pr((i,k) \text{ is eligible})=\Omega(1)$.
\end{lemma}

\begin{proof}
    There are ${n\choose n/2}$ choices for $B(k,\cdot)$
    and ${n/2\choose n/4+t}^2$ ways for $B(k,\cdot)$ 
    to agree with $T(i,\cdot)$ in exactly $n/2+2t$ 
    positions.  By Stirling's approximation
    $N! \sim \sqrt{2\pi N}(N/e)^N$, 
    with the leading constant in $[e^{\frac{1}{12n+1}},e^{\frac{1}{12n}}]$.  Thus, we have
\begin{align*}
\frac{{n/2\choose n/4+t}^2}{{n\choose n/2}}
&\sim \frac{2\pi(n/2)\cdot 2\pi(n/2)}{2\pi(\frac{n}{4}+t)2\pi(\frac{n}{4}-t)\sqrt{2\pi n}}
\frac{((n/2)/e)^{4(n/2)}}{((n/4 +t)/e)^{2(n/4+t)}((n/4-t)/e)^{2(n/4-t)}(n/e)^n}\\
&\geq \frac{4}{\sqrt{2\pi n}}
\exp(2n\ln(n/2) - (n/2+2t)\ln(n/4+t) - (n/2-2t)\ln(n/4-t) - n\ln n)\\
&= \frac{4}{\sqrt{2\pi n}}\exp(-(n/2)[(1+\epsilon)\ln(1+\epsilon) + (1-\epsilon)\ln(1-\epsilon)])\\
\intertext{where $\epsilon = 4t/n$. 
When $x\in (0,1)$ we have 
$\ln(1+x)\leq x$ and $\ln(1-x)\leq -x$. Continuing,}
&\geq \frac{4}{\sqrt{2\pi n}}\exp(-(n/2)2\epsilon^2) \;=\; \frac{4}{\sqrt{2\pi n}}\exp(-16t^2/n).
\end{align*}
In conclusion, 
$\Pr((i,k) \text{ is eligible}) \geq \sum_{t=\sqrt{n}/2}^\infty \frac{4}{\sqrt{2\pi n}}\exp(-16t^2/n) = \Omega(1)$.
\end{proof}

\begin{lemma}\label{lem:AB-to-graph}
For each $T\in\mathcal{T}$, there is a graph
$G[T] = (X\cup Y\cup Z,E)$ where 
$E\subset (X\times Y)\cup (Y\times Z)$
and
$X=(x_i)_{i\in[n]}, Y=(y_i)_{i\in [cn]}, Z = (z_i)_{i\in [n]}$
that approximately
encodes $T$ in the following way.
Let $\tilde{T}$ be defined by the vertex connectivities on $X\times Z$:
\[
\tilde{T}(i,j) = \left\{
\begin{array}{ll}
0 & \mbox{if $\kappa(x_i,z_j) < 4n$}\\
1 & \mbox{if $\kappa(x_i,z_j) \geq 4n$.}
\end{array}
\right.
\]  
Then $T,\tilde{T}$ have Hamming 
distance less than $n^2/12$.
\end{lemma}

\begin{proof}
Roughly speaking we take $A,B$ from \cref{lem:AB-construction} and let 
$A$ be the adjacency matrix
for $E\cap (X\times Y)$ and 
$B$ be the adjacency matrix
for $E\cap (Y\times Z)$.
However for technical reasons we slightly 
modify
the construction of $A$.
In particular, we set $A(i,k)=1$ for
\begin{enumerate}
    \item[(i)] $r=\Theta(c\log n)$ values of 
$k$ chosen uniformly at random, then 
    \item[(ii)] $4n-r$ additional $k$ for which 
$(i,k)$ is eligible.
\end{enumerate}
Note that if $A(i,k)=1$ due to (i), 
$\E(A(i,k)\cdot B(k,j)=T(i,j))=1/2$
(rather than $1/2 \pm 1/\sqrt{n}$), 
which distorts $\E(C(i,j))$ 
by $r/\sqrt{n} < 1$.  
This is within the margin of 
uncertainty for $\tau = 2n\pm O(1)$
tolerated by \cref{lem:AB-construction},
so all the conclusions of \cref{lem:AB-construction} still hold.

Recall that $C=AB$.  Observe that 
$C(i,j)$ reflects the maximum flow from 
$x_i$ to $z_j$ in $G[T]$,
if $G[T]$ were a unit-capacity network
with all edges directed from $X\rightarrow Y\rightarrow Z$.  However, $G[T]$ is \emph{undirected}, and by Menger's theorem,
$\kappa(x_i,z_j)$ is the maximum size set of internally vertex disjoint paths 
connecting $x_i$ and $z_j$.

We claim that with high probability, 
the vertex connectivity $\kappa(x_i,z_j)$ 
is exactly 
\[
\kappa(x_i,z_j) = \min\{C(i,j) + 2(n-1), 4n\}.
\]
In particular we claim that 
$\kappa(x_i,z_j)=|P_{0}\cup P_{1}\cup P_{2}|$, where 
$P_0$ is a set of $C(i,j)$ length-2 paths, 
$P_1$ is a set of $n-1$ paths internally disjoint from $P_0$ of the form 
$x_i \rightarrow Y\rightarrow X\rightarrow Y\rightarrow z_j$ and 
$P_2$ is a set of at most $n-1$ paths 
internally disjoint from $P_0\cup P_1$ 
of the form 
$x_i \rightarrow Y\rightarrow Z\rightarrow Y\rightarrow z_j$.  
See \cref{fig:XYZ}.

\begin{figure}
\centering
    \begin{tabular}{c@{\hspace{1cm}}c}
    \scalebox{.47}{\includegraphics{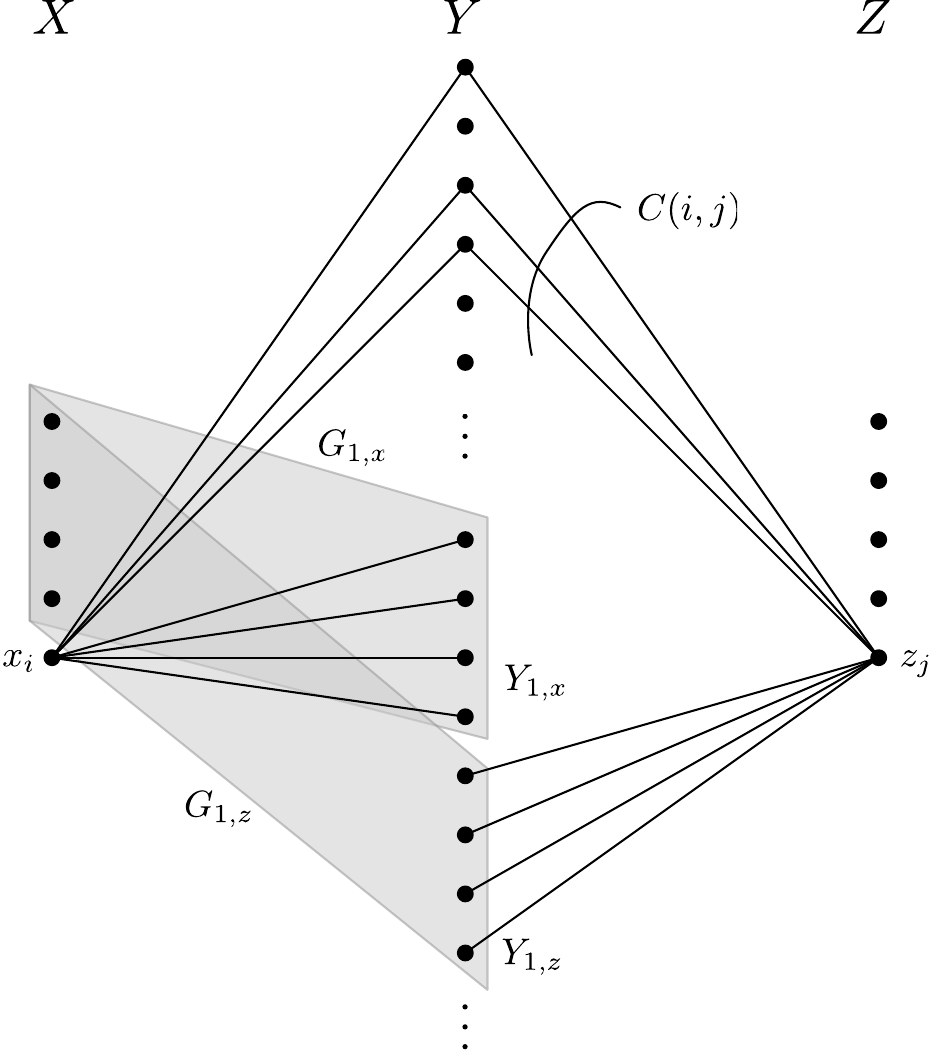}}
&    \scalebox{.47}{\includegraphics{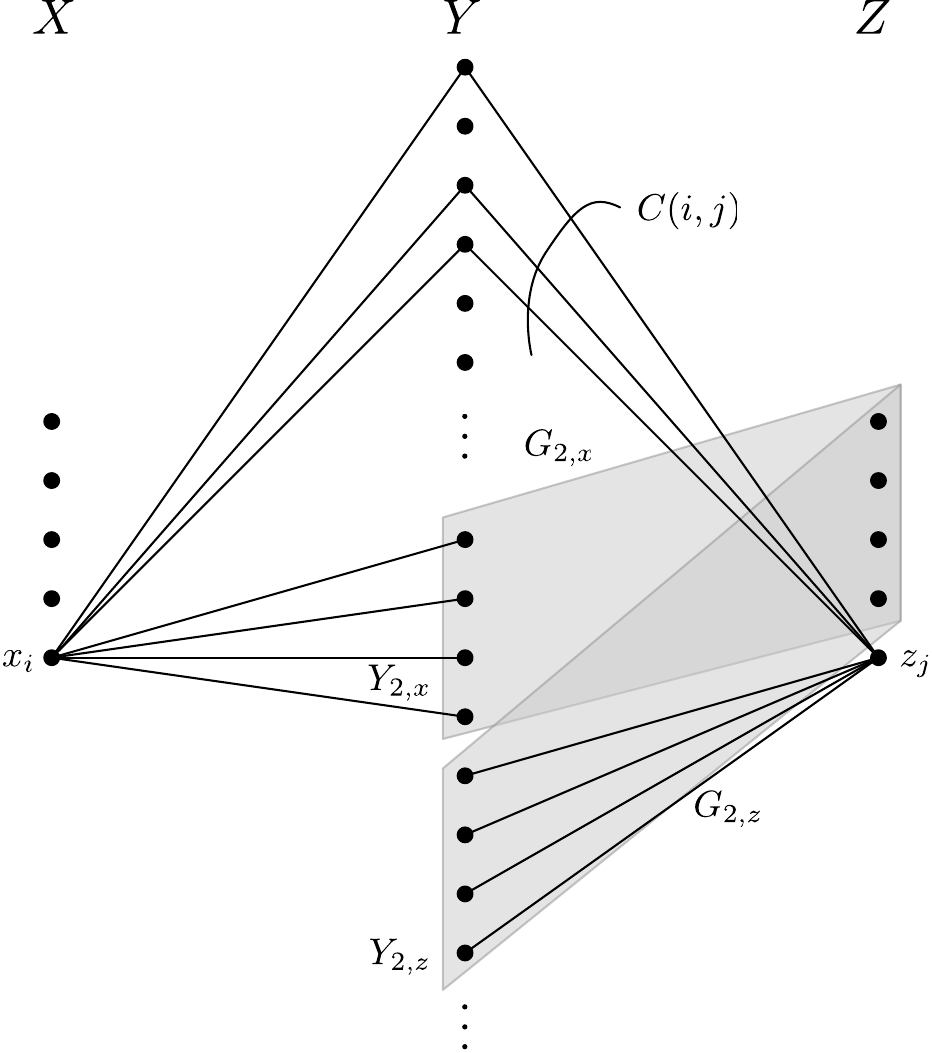}}\\
(a) & (b)
    \end{tabular}
\caption{\label{fig:XYZ} $P_0$ consists of $C(i,j)$ length-2 paths from $x_i$ to $z_j$.  
{\bf (a)} $G_{1,x},G_{1,z}$ contain perfect matchings $M_{1,x},M_{1,z}$ with high probability.  $P_1$ consists of $|Y_{1,x}|=n-1$ paths via $M_{1,x},M_{1,z}$.
{\bf (b)} $G_{2,x},G_{2,z}$ contain perfect matchings $M_{1,x},M_{1,z}$ with high probability.  
$P_2$ consists of $|Y_{2,x}|=\min\{n-1,4n-|P_0|-|P_1|\}$ 
paths via $M_{2,x},M_{2,z}$.}
\end{figure}

We construct $P_1$ as follows.
Choose $Y_{1,x}$ to be any $n-1$ neighbors of $x_i$ that are not part of $P_0$ paths,
and $Y_{1,z}$ to be any $n-1$ neighbors of $z_j$ that are not part of $P_0$ paths.  Clearly $Y_{1,x}\cap Y_{1,z}=\emptyset$.
Let $G_{1,x},G_{1,z}$ be the subgraphs of $G[T]$ induced by $X-\{x_i\}\cup Y_{1,x}$ and $X-\{x_i\}\cup Y_{1,z}$, respectively.  
These graphs contain many random edges of $A$ generated by (i).
It is well known~\cite{Bollobas11,FriezeK16} that such random 
graphs contain perfect matchings with high probability.  
We give a self-contained proof in \cref{lem:perfect-matching}.
Let $M_{1,x},M_{1,z}$ be the perfect matchings of $G_{1,x},G_{1,z}$.
Together with $x_i,z_j$, these 
form $n-1$ paths $P_1$ internally vertex 
disjoint from $P_0$.  
The construction of paths $P_2$ is mostly symmetric.
Let $Y_{2,x}$ be $\min\{n-1, 4n - |P_0| - |P_1|\}$ neighbors of $x_i$ not already included in $P_0,P_1$,
and $Y_{2,z}$ be $|Y_{2,x}|$ neighbors of $z_j$ not included in $P_0,P_1$. Note that both $Y_{2,x}$ and $Y_{2,z}$ have size at least  $0.9n$ because $|P_0| \le 2.1n$ and $|P_1|=n-1$.
By \cref{lem:perfect-matching}, the graphs 
$G_{2,x},G_{2,z}$ induced by $Z-\{z_j\}\cup Y_{2,x}$ and $Z-\{z_j\}\cup Y_{2,z}$ also contain perfect matchings
$M_{2,x},M_{2,z}$ with high probability.
Together with $x_i$ and $z_j$ these generate $|Y_{2,x}|$ paths internally vertex disjoint from $P_0,P_1$.

The construction of $P_0\cup P_1\cup P_2$ shows that
$\kappa(x_i,z_j) \geq \min\{C(i,j)+2(n-1),4n\}$.
The upper bound $\kappa(x_i,z_j) \leq \min\{C(i,j)+2(n-1),4n\}$
follows from taking the better of two vertex cuts.
We can separate $x_i$ from $z_j$ by removing $x_i$'s neighborhood
of $4n$ vertices, or we can remove the $|C(i,j)|$ 
$Y$-vertices on the $P_0$ paths, together with 
$X\cup Z - \{x_i,z_j\}$.

Suppose that $T(i,j)=1$.
By \cref{lem:AB-construction}, 
$\E(C(i,j))\geq 2n+4\sqrt{n}$ and with probability $(1+o(1))\exp(-8)$, 
$C(i,j)\geq 2n+2$ and $\kappa(x_i,z_j)=4n$.
If $T(i,j)=0$ then 
$\E(C(i,j))\leq 2n-4\sqrt{n}$ and with
probability $(1+o(1))\exp(-8)$,
$C(i,j) < 2n+2$ and $\kappa(x_i,z_j)<4n$.
Thus, in expectation $\tilde{T},T$ have Hamming distance much less than $n^2/12$,
so there must exist a graph $G[T]$ having this property.
\end{proof}

\begin{lemma}\label{lem:perfect-matching}
The graphs $G_{1,x},G_{1,z},G_{2,x},G_{2,z}$ have perfect matchings with high probability.
\end{lemma}

\begin{proof}
We claim the degree of any $x_j \in X-\{x_i\}$ 
in $G_{1,x}$ or $G_{1,z}$ is $\Omega(\log n)$
with high probability.  Due to step (i) in the construction of $A$, $x_j$ picks $\Theta(c\log n)$ neighbors in $Y$ uniformly at random,
each of which lies in $Y_{1,x}$ with probability $\Theta(1/c)$.  By a Chernoff-Hoeffding bound,
$x_j$ has at least $\Omega(\log n)$ neighbors in $Y_{1,x}$ with probability $1-1/\poly(n)$.

In the remainder of the proof, we only
assume $G'=(L\cup R,E)$ is a random bipartite graph with $|L|=|R|$ where for $v\in L$, 
$\deg(v)=d$ and $N(v)$ is a uniformly random 
$d$-subset of $R$, 
$d \geq c_0 \ln n \geq 3\ln n$.
(By construction every row of $B$ has 
Hamming weight $n/2$, which implies that in 
$G'=G_{2,x}$, $d=n/2$, and in $G'=G_{2,z}$, $d=n/2-1$.)

By Hall's theorem, if $G'$ has no perfect matching then there exists a \emph{witness} 
$(I,J)$ with $I\subset L$, $J\subset R$, 
$N(I)\subseteq J$ and $|J|=|I|-1$.
We have 
\[
\Pr(N(I)\subseteq J) = \Pr\left(\bigcup_{v\in I} N(v) \subset J\right) \leq (|J|/|R|)^{d|I|}.
\]
Letting $n=|L|=|R|$ and $\ell=|I|$,
the expected number of such witnesses
is $\sum_{\ell = d+1}^n {n\choose \ell}{n\choose 
\ell-1}\left(\frac{\ell-1}{n}\right)^{d\ell}$.
We bound the sum in two parts: $\ell\leq \ceil{n/2}$ 
and $\ell > \ceil{n/2}$.
\begin{align*}
    \sum_{\ell = d+1}^{\ceil{n/2}} {n\choose \ell}{n\choose \ell-1}\left(\frac{\ell-1}{n}\right)^{d\ell}
    &\leq \sum_{\ell=d+1}^{\ceil{n/2}} \left(\frac{en}{\ell}\right)^{2\ell}\left(\frac{\ell-1}{n}\right)^{d\ell} < n^{-\omega(1)}.
\intertext{Now let $\ell'=n-\ell$.  The second part is then}
    \sum_{\ell'=0}^{\floor{n/2}} {n\choose \ell'}{n\choose \ell'+1}\left(1-\frac{\ell'+1}{n}\right)^{d(n-\ell')}
    &\leq 
    \sum_{\ell'=0}^{\ceil{n/2}} \exp\left(2(\ell'+1)\ln\left(\frac{en}{\ell'+1}\right)-\frac{d(\ell'+1)(n-\ell')}{n}\right)\\
    &< O(n^{-(c_0-2)}).
\end{align*}
Thus, with probability $1-O(n^{-(c_0-2)})$, there are no Hall
witnesses and $G'$ has a perfect matching.
\end{proof}

\cref{thm:lb} implies strong lower bounds on a variety of data structures for representing vertex connectivity information.

\begin{theorem}\label{thm:vertex-connectivity-lb}
Suppose $\mathscr{D}(G,k)$ is a data structure for 
an undirected $n$-vertex 
graph $G$ that can answer 
any of the following query types:
\begin{description}
\item[Threshold Queries.] Determine whether 
$\kappa(u,v) < k$ or $\kappa(u,v) \geq k$.
\item[Equality Queries.] Determine whether 
$\kappa(u,v) = k$ or $\kappa(u,v) \neq k$.
\item[Gap Queries.]
Distinguish 
$\kappa(u,v) \leq k-\sqrt{k}/4$ from 
$\kappa(u,v) \geq k + \sqrt{k}/4$.
(If $\kappa(u,v)\in (k-\sqrt{k}/4,k+\sqrt{k}/4)$ then 
any answer is correct.)
\end{description}
Then $\mathscr{D}$ requires $\Omega(kn)$ bits of 
space in the worst case.
\end{theorem}

\begin{proof}
Pick $T\in \mathcal{T}$, 
and let $G[T]$ be the $(c+2)n$-vertex graph 
encoding of $T$ from \cref{thm:lb}.
Using $\mathscr{D}(G[T], k)$ with $k=4n$,
we can generate the matrix $\tilde{T}$
with \emph{Threshold} queries or \emph{Equality} queries
(since no $\kappa(x_i,z_j)$-values are 
strictly greater than $4n$)
then deduce $T$ since $\tilde{T}$ is within its decoding radius $n^2/12$.
Thus, $\mathscr{D}$ requires at 
least $\log|\mathcal{T}| = \Theta(n^2)$ bits of space.

Turning to \emph{Gap} queries, 
take $k=4n-\sqrt{n}/2$, so the data structure correctly 
distinguishes
$\kappa(x_i,z_j) \leq k - \sqrt{k}/4 = 4n-\sqrt{n}+O(1)$
from
$\kappa(x_i,z_j) \geq k + \sqrt{k}/4 = 4n-O(1)$.
We generate $\tilde{T}$ according to these queries.
When $T(i,j)=1$, $\kappa(x_i,z_j)=4n$ with probability $1-(1+o(1))\exp(-8)$ and in this case $\tilde{T}(i,j)=1$.
When $T(i,j)=0$, $\kappa(x_i,z_j)$ is less than 
$4n-4\sqrt{n}$ in expectation. 
Recall that $\kappa(x_i,z_j) = C(i,j) + |P_1|+|P_2|$, 
where the number of 4-edge paths $|P_1|+|P_2|$ is at most $2(n-1)$.  Hence, $\kappa(x_i,z_j) \geq 4n-\sqrt{n}$ implies that
$C(i,j) \geq 2n-\sqrt{n}+2$. 
Recall that $C(i,j)$ is the sum
of independent Bernoulli random variables, 
making it susceptible to Chernoff-Hoeffding bounds.
We have $\E(C(i,j)) < 2n - 4\sqrt{n}+1$,\footnote{Recall 
that in the modified construction of $A$ from \cref{lem:AB-to-graph},
$C(i,j)$ is the sum of $4n$ Bernoulli random variables, 
$r=O(c\log n)$ with expectation 1/2, 
and $4n-r$ with expectation at most $1/2-1/\sqrt{n}$,
so $\E(C(i,j)) < 2n-4\sqrt{n}+1$.}
which implies the probability of
$\kappa(x_i,z_j) \geq 4n-\sqrt{n}$ is at most 
$\exp(-2(3\sqrt{n})^2/4n) = \exp(-4.5) < 1/12$.
Thus, in expectation $\tilde{T}$ is still within 
the decoding radius of $T$.

For general values of $k$, we take $G$ to be
the disjoint union $G[T_1]\cup \cdots \cup G[T_{n/k}]$,
where $T_1,\ldots,T_{n/k}\in \mathcal{T}$
and each $G[T_i]$ has $\Theta(k)$ vertices.
There are $(2^{\Theta(k^2)})^{n/k}=2^{\Theta(kn)}$ choices
of $G$, and since we can decode $T_1,\ldots,T_{n/k}$ with
\emph{Threshold}/\emph{Equality}/\emph{Gap} queries,
the space for $\mathcal{D}(G,k)$ must be $\Omega(kn)$ bits.
\end{proof}

\begin{corollary}\label{cor:approximation-lb}
Fix any $\epsilon = \Omega(1/\sqrt{n})$.  
Suppose $\mathscr{D}(G,\epsilon)$ is a data structure 
for an undirected $n$-vertex graph $G$ that reports 
$\tilde{\kappa}(u,v) \in [(1-\epsilon)\kappa(u,v),(1+\epsilon)\kappa(u,v)]$.
Then $\mathscr{D}$ requires 
$\Omega(n\epsilon^{-2})$ bits of space.
\end{corollary}

\begin{proof}
The input graph $G$ is the disjoint
union of $G[T_1]\cup\cdots \cup G[T_{n/k}]$,
where $T_1,\ldots,T_{n/k}\in \mathcal{T}$,
each $G[T_i]$ has $\Theta(k)$ vertices, 
and $k=\Theta(\epsilon^{-2})$.  
Note that when $u,v\in V(G[T_i])$,
a $(1+\epsilon)$-approximation $\tilde{\kappa}(u,v)$
can be used to answer a \emph{Gap} query with respect to $k$.
Thus, by \cref{thm:vertex-connectivity-lb},
$\mathscr{D}(G,\epsilon)$ also requires $\Omega(kn)=\Omega(\epsilon^{-2}n)$ bits of space.
\end{proof}

\section{Vertex Connectivity Oracles}\label{sect:vconn-oracle}

\cref{lem:relation_between_connectivities} 
says that for any terminal set $U$, 
$\kappa'_{G,U}(u,v)$ \emph{never} underestimates the true 
value $\kappa_G(u,v)$ and achieves the equality $\kappa'_{G,U}(u,v)=\kappa_G(u,v)$ if $U$ 
\emph{captures} $u, v$.
Following~\cite{ChuzhoyK12}, 
Izsak and Nutov's~\cite{IzsakN12} ingenious algorithm 
proceeds by sampling several terminal sets 
$U_1,U_2,\ldots \subseteq V$, 
including each vertex with probability $1/k$. 
Each terminal set includes both 
$u$ and $v$ with probability $1/k^2$ and \emph{avoids} 
a $u$-$v$ cut 
$C_{u,v}$ with constant probability if $|C_{u,v}| \leq k$,
i.e., each terminal set $U_i$ \emph{captures} $u,v$ with probability $\Theta(1/k^2)$,
hence $O(k^2\log n)$ terminal sets suffice 
to accurately capture \emph{all} 
vertex connectivities up to $k$, with high probability. 
The space for the centralized data structure is 
just that of storing $O(k^2\log n)$ Gomory-Hu trees and data structures for answering bottleneck-edge queries. 
Each tree is on $O(n/k)$ terminals in expectation, 
for a total of $O(kn\log n)$ space.
A query $\vconn(u,v)$ needs to examine all terminal sets
containing \emph{both} $u$ and $v$.  This is a classic \textsf{SetIntersection} query.  
Each of $u,v$ is in $\Theta(k\log n)$ sets, 
but are jointly in just $\Theta(\log n)$ sets.

In this section we show how to build random \textsf{SetIntersection} instances 
that are nonetheless \emph{structured}
so that queries can be answered optimally in 
time linear in the output size.
(If minimum cuts are associated with edges in the Gomory-Hu tree, a $\vcut(u,v)$ query can be answered in $O(1)$ additional time be returning a pointer to the appropriate 
cut.  This increases the space to $O(k^2n\log n)$.)

First, we show how to find $O(k^2\log n)$ terminal sets that capture all pairs in $V^2$ w.h.p. using a
construction based on affine planes and 
3-wise independent hash functions. 

\begin{lemma}\label{lem:generating_terminal_set}
Given an error parameter $\delta>0$ and a set $S\subseteq V$ of vertices, 
there is an algorithm using $O(\log n\log\delta^{-1})$ 
random bits that generates 
a terminal set family $\mathcal{U}$ with the following properties.
\begin{itemize}
    \item $|\mathcal{U}|=O(k^2\log \delta^{-1})$ and each $U\in\mathcal{U}$ has $|U|=O(|S|/k)$.
    \item Given vertices $u,v\in S$ 
    with $\kappa(u,v)\leq k$ 
    we can find $O(\log \delta^{-1})$ sets
    $U_1,\ldots,U_{O(\log \delta^{-1})}$ in $\mathcal{U}$ in
    $O(\log \delta^{-1})$ time, such that
    each, independently, 
    captures $u,v$ with constant probability.  
    As a consequence, 
    $\mathcal{U}$ captures all of $S^2$
    with probability $1-\delta|S|^2$.
\end{itemize}
\end{lemma}

\begin{proof}
We identify vertices in $S$ with the integers $[|S|]=\{0,1,\ldots,|S|-1\}$.
If $\abs{S} = O(k)$ then a trivial construction $\mathcal{U}=\{\{u, v\}: u, v \in S\}$ satisfies all the requirements. 
We assume in the following that  $\abs{S} \ge 100k$.

Let $p_0$ be the first prime larger than $|S|$
and $p$ be the first prime larger than $2k$. 
Let $H = (P,L)$ be 
a subset of an affine plane, defined as follows.
$P=\{u_{i,j} \mid i\in [\ceil{p_0/p}], j\in [p]\}$ is a set of points
arranged in a rectangular grid
and $L = \{\ell_{s,j} \mid s,j\in [p]\}$ is a set of lines, where
$\ell_{s,j} = \{u_{t,j+st \mod p} \mid t\in [\ceil{p_0/p}]\}$ is the
line with slope $s$ passing through $u_{0,j}$.
Choose a hash function $\Bar{h}$ from a 3-wise independent family:
\[
\Bar{h}(x)=(ax^2 + bx + c) \mod p_0,
\]
where $a,b,c$ are chosen uniformly at random from $[p_0]$. 
The mapping $h : S\rightarrow P$ from $S$ to points is given by:
\[
h(x) = u_{i,j}, 
\text{ where } i = \floor{\Bar{h}(x) / p} 
\text{ and } 
j = \Bar{h}(x)\mod p.
\]
Form the $p^2=\Theta(k^2)$ 
terminal sets $\mathcal{U}[h] = \{U_{s,j} \mid s,j\in [p]\}$ as follows.
\[
U_{s,j} = \{v \in S \mid h(v) \in \ell_{s,j}\}.
\]
Now fix any two vertices $u,v\in S$ 
whose minimum cut $C$ has size $\kappa(u,v)\leq k$.  
Then one set in 
$\mathcal{U}[h]$ will capture $u,v$ if
(i) $h(u),h(v)$ differ in their first coordinates of the rectangular grid $P$, and, assuming this happens,
(ii)  $C\cap U_{s,j} = \emptyset$ where $\ell_{s,j}$ is the unique line containing $h(u),h(v)$.
The probability of (i) 
is $1-p/p_0 \ge 9/10$ and the probability of (ii) is
\begin{align*}
  \Pr(C\cap U_{s,j}=\emptyset \mid h(u),h(v)) 
  &\geq 1 - \sum_{x\in C\cap S} \Pr(h(x) \in \ell_{s,j} \mid h(u),h(v))\\
  &\geq  1 - k\ceil{p_0/p}/p_0 \geq 0.49,
\end{align*}
where the first inequality is by the union bound and the 
second inequality follows from 
$|\ell_{s,j}| \le \ceil{p_0/p}$, $|C| \le k$.
We sample $h$ from a $3$-wise 
independent hash family
since we need $h(x)$ to be uniformly random after 
conditioning on the values of 
$h(u)$ and $h(v)$.

Let $\mathcal{U}$ be the union of $\mathcal{U}[h_1],\ldots,\mathcal{U}[h_{O(\log \delta^{-1})}]$ using independent hash functions $h_1,\ldots,h_{O(\log \delta^{-1})}$.
Then it follows that $|\mathcal{U}| = O(p^2\log\delta^{-1}) = O(k^2\log\delta^{-1})$.
Suppose we preprocess a table of inverses modulo $p$.
Given $u,v$ with $\kappa(u, v)\leq k$,
we can clearly identify whether points $h_i(u), h_i(v)$ differ in their first coordinates, 
and if so, identify $\ell_{s,j}$ (and hence $U_{s, j}\in \mathcal{U}[h_i]$) containing them 
both in $O(1)$ time.
Therefore, we can identify 
$O(\log\delta^{-1})$ members of $\mathcal{U}$
in $O(\log \delta^{-1})$ time such that at least one 
of these terminal sets captures $u,v$ with probability 
$1-\delta^{-1}$.

We now turn to the load balancing condition $|U|=O(\abs{S}/k)$.
Note that when $a\neq 0$, $h$ is a 2-to-1 function
as the polynomial defining $h$ has degree 2.
When $a=0,b\neq 0$, $h$ is 1-to-1.
Thus, whenever $(a,b,c)\neq (0,0,c)$,
every $U\in \mathcal{U}$ has size 
$|U|<2\ceil{p_0/p} = O(\abs{S}/k)$.  
The performance of the algorithm 
is clearly quite bad when $a=b=0$ 
as each terminal set $U$ is either $\emptyset$ or $S$,
so we can remove these hash functions from the hash 
family and only improve the probability of success.
\end{proof}

\cref{thm:vconn-oracle} 
is obtained by applying \cref{lem:generating_terminal_set}
with $S=V$ to the Izsak-Nutov oracle~\cite{IzsakN12}.
The claimed construction time will be substantiated
in \cref{sec:GH_tree_construction}.

\begin{theorem}\label{thm:vconn-oracle}
Let $G=(V,E)$ be an undirected graph, $k\in[1,n]$,
and $\bar{m}=\min\{m,(k+1)n\}$.
A data structure with size $O(kn\log n)$ can be constructed
in $O(m)+O(k^3\log^7n)\cdot T_{\flow}(\bar{m}) = O(m+k^3\bar{m}^{1+o(1)})$ time such that
a $\vconn(u,v)=\min\{\kappa_G(u,v),k+1\}$ query
can be answered in $O(\log n)$ time, which is correct with high probability.
Using space $O(k^2n\log n)$, a $\vcut(u,v)$ query can be answered in $O(1)$ \emph{additional} time; 
it returns a pointer to a $\kappa_G(u,v)$-size 
$u$-$v$ cut whenever $\kappa_G(u,v)\leq k$.
\end{theorem}

In \cref{thm:vconn-oracle}, 
when $m \gg kn$, the Nagamochi-Ibaraki~\cite{NagamochiI92} 
algorithm reduces the number of edges 
from $m$ to $\bar{m}=(k+1)n$ so that all $k'$-cuts, 
$k'\leq k$, are the same 
in the original and reduced graph.
It runs in $O(m)$ time.
On the other hand, when $m\ll kn$, the output 
of \cref{thm:vconn-oracle} is \emph{still} 
$\Theta(kn\log n)$, i.e., it does not benefit 
from the \emph{sparseness} of the input.
\cref{thm:vconn-oracle-sparse} gives a smaller
vertex connectivity oracle when the input graph is sparse.  It is obtained by applying \cref{lem:generating_terminal_set} to several
terminal sets.

\begin{theorem}\label{thm:vconn-oracle-sparse}
Let $G=(V,E)$ be an undirected graph, $k\in[1,n]$,
and $\bar{m}=\min\{m,(k+1)n\}$.
A data structure with size $O(m\log k\log n)$ can be constructed
in $O(m)+O(k^3\log^7n)\cdot T_{\flow}(\bar{m}) = O(m+k^3\bar{m}^{1+o(1)})$ 
time such that a $\vconn(u,v)=\min\{\kappa_G(u,v),k+1\}$ query
can be answered in $O(\log n)$ time, which is correct with high probability.
Using space $O(k^2n\log n)$, a $\vcut(u,v)$ query can be answered in $O(1)$ \emph{additional} time; it returns a pointer to a $\kappa_G(u,v)$-size 
$u$-$v$ cut whenever $\kappa_G(u,v)\leq k$.
\end{theorem}

\begin{proof}
Let $S_g \subseteq V$ be the set of all $u\in V$ with 
$\deg(u) \geq 2^g$.  Clearly $|S_g|=O(m/2^g)$.
Using \cref{lem:generating_terminal_set} with $S=S_g$ and $k'=2^{g+1}$, 
we can build a data structure with size $O(k'|S_g|\log n) = O(m\log n)$
that answers $\min\{\kappa(u,v),k'+1\}$-queries for any $u,v\in S_g$, 
in $O(\log n)$ time, that are correct with high probability.
We build such a data structure for each $g\in [0,\floor{\log_2 k}]$.  
Given a query pair $u,v$ we calculate 
$g=\min\{\floor{\log_2 k}, \floor{\log_2(\min\{\deg(u),\deg(v)\})}\}$
and query the data structure on $S_g$.  
It returns $\min\{\kappa(u,v), 2^{g+1}+1\}$ with high probability.
When $g<\log_2 k$ this is exactly $\kappa(u,v)$, 
since $\kappa(u,v)\leq \min\{\deg(u),\deg(v)\} < 2^{g+1}$.
When $g=\log_2 k$ this is just as good as a 
$\min\{\kappa(u,v), k+1\}$ query.
\end{proof}

\section{Gomory-Hu Trees for Element Connectivity}\label{sec:GH_tree_construction}

The goal in this section is to prove the following:

\begin{theorem}\label{thm:construct_k_gh_tree}
A $k$-Gomory-Hu tree for element connectivity w.r.t.~an $m$-edge graph $G$ and terminal set $U$ can be constructed
in $O(k\log^6n)\cdot T_{\flow}(m)$ time.
\end{theorem}

Observe that the construction time claims of \Cref{thm:vconn-oracle} and \Cref{thm:vconn-oracle-sparse} 
follow from~\Cref{thm:construct_k_gh_tree}
because there are $O(k^2 \log n)$ terminal sets,
each of size $|U|=O(n/k)$.

\paragraph{Obstacles to Adapting Algorithms of \cite{Li2021Approximate} for Element Connectivity.}
The proof of \Cref{thm:construct_k_gh_tree} is obtained by adapting the Gomory-Hu tree construction for \emph{edge connectivity} by Li and Panigrahi \cite{Li2021Approximate} to work for \emph{element connectivity}. 
Although we use the same high-level approach, element connectivity introduces some extra complications that we need to deal with.

For example, given an input graph $G$, all Gomory-Hu tree algorithms for edge connectivity proceed by finding a minimum edge cut $(A,B)$, contracting one side, say $B$, of the cut into a single vertex $b$, and recursing 
on the contracted graph, denoted $G'$. 
By submodularity of edge cuts, we have that the edge connectivity between any two vertices $a_1,a_2\in A$ are preserved in $G'$. This is crucial for the correctness of the whole algorithm.

Unfortunately, the direct analogue of this statement fails for element connectivity.
For example, suppose $p, q\notin B$ are disconnected by a mixed cut $C$ of graph $G$. Then in graph $G'$, $C'=\{b\}\cup (C- B)$ becomes a mixed cut disconnecting $p$ and $q$. As long as $C$ contains more than one element in $B$ (an edge, or a non-terminal vertex), $\abs{C'} < \abs{C}$, 
so $\kappa'(p, q)$ is \emph{smaller} in the contracted graph $G'$.

Our solution is to exploit the generality of element connectivity. 
When we recurse on a contracted graph, 
we add the contracted node as a \emph{de facto} terminal, i.e., it cannot appear in any mixed cut. It will be shown in \cref{lem:elem_conn_preserved}
and \cref{lem: elem_conn_preserved_for_small_parts} that the recursion becomes valid with this modification. However, unlike normal terminals, 
we cannot try to support connectivity queries on these new terminals,
so they will be passed along in the recursion as a \emph{forbidden set}
$F$.  This is the motivation for \cref{lem:isolating_cut_lemma_with_forbidden_terminals}.

In the rest of this section, we formally prove \Cref{thm:construct_k_gh_tree}.

\subsection{Approximate Gomory-Hu Trees for 
Element Connectivity}
\label{subsec:algorithm-gh-tree-description}

Instead of proving \Cref{thm:construct_k_gh_tree} directly, 
it is more convenient to prove the following 
$(1+\epsilon_0)$-approximation analogue of \cite{Li2021Approximate} for element connectivity.  By setting $\epsilon_0 = 1/(k+1)$, 
the resulting Gomory-Hu tree is exact for element
connectivities up to $k$.

\begin{theorem}\label{thm:construct_approx_gh_tree}
For $0 < \epsilon_0 < 1$, 
a $(1+\epsilon_0)$-approximate Gomory-Hu tree for element connectivity w.r.t. graph $G$ and terminal set $U$ can be computed in $O(\epsilon_0^{-1}\log^6n)\cdot T_{\flow}(m)$ time, where $T_{\flow}(m)$ is the time to compute max flow 
in an $m$-edge graph.
\end{theorem}

Before we give the proof of \cref{thm:construct_approx_gh_tree}, 
observe that as a simplifying assumption, \cref{lem:isolating_cut_lemma_with_forbidden_terminals} 
(isolating cuts with forbidden terminals)
required that $I\cup F$ be an independent set.
We can force any instance to satisfy this property
by subdividing all edges in $E\cap {I\cup F\choose 2}$.
As a consequence, throughout this section, we assume that all mixed cuts
in the modified graph consist solely of (non-terminal) vertices.\footnote{
\label{footnote:solely-vertex-assumption}
If $\{x,y\}$
is subdivided into 
$\{x,v_{x,y}\},
\{v_{x,y},y\}$, 
then any vertex-only mixed cut
containing $v_{x,y}$ in the modified instance contains edge $\{x,y\}$ in the original instance, and vice-versa.} This introduces at most 
$\abs{E\cap \binom{U}{2}} = O(m)$ new edges 
and does not affect the scale of the running time.

The algorithm for \cref{thm:construct_approx_gh_tree} is 
\textsc{ApproxElemConnGHTree} (\cref{algo:construct_gh_tree}).
It takes a graph $G$, terminal set $U\subseteq V(G)$,
forbidden set $F\subseteq V(G)$, and approximation parameter $\epsilon$.  
It returns a $(1+\epsilon)^{\log \abs{U}}$-approximate Gomory-Hu tree $(T,g,C)$ for element connectivity, 
where $g:(U\cup F)\rightarrow V(T)$, and $T$ 
correctly represents minimum vertex cut sizes 
avoiding $U\cup F$, for any two terminals in $U$.
$C$ maintains the cuts associated with the tree $T$ to support $\vcut$ queries. 

The algorithm is recursive.
We ensure $\lambda$ is a lower bound on the minimum global 
element connectivity by initializing it at $1$ and increasing it 
whenever we are confident it has changed.\footnote{Specifically, if 
$\alpha$ and $\beta$ are two calls 
in the recursion tree, and the path from $\alpha$ to $\beta$ includes
$\Delta = \Theta(\log^3n)$ calls
in the ``$G_{\LG}$'' branch of the recursion, then 
we are confident that
the minimum element cut has increased
by at least a $1+\epsilon$ factor from $\alpha$ to $\beta$, w.h.p.
Thus, we will set $\lambda \gets (1+\epsilon)\lambda$ whenever
the $G_{\LG}$-depth of the current recursive call is a multiple of $\Delta$. See \cref{subsec:proof_of_thm_51} for more detail.}
Let $R^j\subset U$ be a random sample of the terminals,
and $\{S_v^j\}_{v\in R^j}$ 
be such that $\partial S_v^j$ is a minimum
cut avoiding $U\cup F$ and isolating
$v$ from $R^j-\{v\}$.
If
\begin{align}
|\partial S_v^j| &\leq (1+\epsilon)\lambda\label{eqn:onepluseps}
\end{align}
then
by definition of $\lambda$, $\partial S_v^j$ must be 
a $(1+\epsilon)$-approximate minimum cut 
for any terminals $u\in S_v^j$, $u'\not\in S_v^j\cup \partial S_v^j$.  Thus, we can split our problem into 
multiple contracted subproblems by recursively 
considering only terminals on one side of the cut.
In order to bound the \emph{depth} of this recursion
it is also desirable that
\begin{align}
    |U\cap S_v^j| &\leq |U|/2.\label{eqn:geometric-shrink}
\end{align}
Let $R_{\SM}^j \subseteq R^j$ be the subset of sampled
terminals that satisfy the desirable features (\ref{eqn:onepluseps}) and (\ref{eqn:geometric-shrink}).
The procedure \textsc{CutThresholdStep} 
(Line~\ref{line:CutThresholdStep} of \textsc{ApproxElemConnGHTree})
computes 
many sampled sets $R^0,\ldots,R^{\log|U|}$ with geometrically decreasing probabilities.  
However, 
only the set $R^i$ maximizing $U\cap \left(\bigcup_{v\in R_{\SM}^i} S_v^i\right)$ is actually used by \textsc{ApproxElemConnGHTree}.  
(For technical reasons, 
\textsc{CutThresholdStep} selects
one ``source'' terminal $s\in U$ 
to appear in all $R^i$ sets, 
which is specifically excluded from every $R^i_{\SM}$.)

Once $R^i,R_{\SM}^i$ and $\{S_v^i\}_{v\in R^i}$ are fixed, 
\textsc{ApproxElemConnGHTree}
proceeds to call itself recursively on
each $G_v$, $v\in R_{\SM}^i$, 
which is derived from $G$ by contracting 
$V(G) - (S_v^i\cup \partial S_v^i)$ 
to a single $F$-vertex 
$x_v$ (Line~\ref{line:rec-call}).
The terminal set in this recursive call is 
naturally $U_v = U\cap S_v^i$.
The remaining graph, which we call $G_{\LG}$,
is obtained by contracting each $S_v^i$, $v\in R_{\SM}^i$, to a single $F$-vertex $y_v$.  
A recursive call is also made on $G_{\LG}$ (Line~\ref{line:LG-rec-call}). Notice that the newly introduced $F$-vertices are not adjacent to any existing old $U\cup F$ vertices. 

The output $(1+\epsilon)$-approximate 
Gomory-Hu tree is obtained by combining 
the approximate 
Gomory-Hu trees of $\{G_v\}_{v\in R_{\SM}^i}$ 
and $G_{\LG}$, called $(T_v, g_v, C_v)$ and $(T_{\LG},g_{\LG},C_{\LG})$.
In particular, we join the vertices 
$g_v(x_v)$ and $g_{\LG}(y_v)$ by an edge 
and identify it with the cut $\partial S_v^i$.
Note that \emph{all} vertices in 
$U\cup F \cup \{x_v,y_v\}_{v\in R_{\SM}^i}$ 
are embedded into some $T_v$ or $T_{\LG}$.
(The base case embedding 
is handled in Line~\ref{line:basecase},
when $|U|=1$, and the general inductive case in
Line~\ref{line:combine-GH-trees}.)

\medskip 

The correctness of \textsc{ApproxElemConnGHTree} 
will be proved in \cref{sect:GH-correctness}
and the running time analysis in
\cref{subsubsec:run_time_analysis}.
The correctness rests on two claims: 
that for any two 
terminals $u,u'\in U_{\LG} = U\cap V(G_{\LG})$, 
their element connectivity w.r.t.~$G,U\cup F$ 
is preserved \emph{exactly} in $G_{\LG}, U_{\LG}\cup F_{\LG}$ (\cref{lem:elem_conn_preserved}), 
and that for any two $u,u'\in U_v = U\cap V(G_v)$,
their element connectivity w.r.t.~$G_v,U_v\cup F_v$ is at most 
$1+\epsilon$ times their element connectivity w.r.t.~$G,U\cup F$ (\cref{lem: elem_conn_preserved_for_small_parts}).
Since this latter case 
can only occur $\log|U|$ times
for two $u,u'\in U$, the final tree is a
{$(1+\epsilon)^{\log|U|}$-approximate} Gomory-Hu tree. By picking $\epsilon = \Theta(\epsilon_0/\log n)$ the approximation ratio becomes $(1+\epsilon)^{\log\abs{U}}\leq 1 + \epsilon_0$.

To summarize, to compute a
$(1+\epsilon_0)$-approximate Gomory-Hu tree for terminal set $U$, we make the initial 
call to \textsc{ApproxElemConnGHTree}  (\cref{algo:construct_gh_tree}) 
with $F=\emptyset$, $\epsilon = \epsilon_0/(2\log \abs{U})$ and $\lambda = 1$. The output $(T, g, C)$ matches with \cref{def:elem_conn_gh_tree} where $f = g$. If only $\vconn$ queries are required, the data structure as described in \cref{sect:vconn-oracle} only stores $(T, g)$.

\begin{algorithm}
\caption{\textsc{ApproxElemConnGHTree}$(G(V, E), U, F, \epsilon, \lambda)$}\label{algo:construct_gh_tree}
\SetKwFunction{ApproxElemConnGHTree}{ApproxElemConnGHTree}
\SetKwFunction{Combine}{Combine}
\SetKwFunction{CutThresholdStep}{CutThresholdStep} %
\SetKwInOut{Input}{input}\SetKwInOut{Output}{output}
\DontPrintSemicolon
\Input{A graph $G=(V, E)$, a terminal set $U\subseteq V$, a forbidden set $F\subseteq V$, an accuracy parameter $\epsilon$, and global element connectivity lower bound $\lambda$.}
\Output{An approximate element-connectivity Gomory-Hu tree $(T, {\algoF}, C)$.  For $u,v\in U$, $T$ encodes a $(1+\epsilon)^{\log\abs{U}}$-approximation of $\kappa'_{G,U\cup F}(u,v)$ and a corresponding cut.}
\If(\tcp*[f]{The Base Case}){$\abs{U}=1$}{\label{line:basecase}
Construct $T$ with one node $t$, $\algoF(u)\gets t$ for all $u\in U\cup F$, and $C$ is vacuous.\;
\Return $(T, \algoF, C)$
}
\If{\emph{this recursive call 
has been in the ``$G_{\LG}$ branch'' for 
$\Delta$ consecutive calls}}{
$\lambda\gets (1+\epsilon)\lambda$\tcp*{$\Delta = \Theta(\log^3n)$; see \cref{subsec:proof_of_thm_51}}
}
Call \CutThresholdStep{$G, U, F, (1+\epsilon)\lambda$} and store its output $s,\{R_{\SM}^j,R^j,\{S_v^j\}_{v\in R_{\SM}^j}\}_j$\;\label{line:CutThresholdStep}
Fix $i\in\{0, 1, \cdots, \lowerint{\log\abs{U}}\}$ that maximizes $\abs{\cup_{v\in R_{\SM}^i}(S_v^i\cap U)}$.\;

\ForEach{$v\in R_{\SM}^i$}{
    $G_{v}\gets$ the graph derived from $G$ 
    with $V- (S_{v}^i\cup \partial S_v^i)$ contracted to $x_{v}$.\; 
    $U_{v}\gets S_v^i\cap U$.\;
    $(T_{v}, {\algoF}_{v}, C_{v})\gets$ \ApproxElemConnGHTree{$G_v, U_v, F\cup\{x_v\}, \epsilon, \lambda$}.\;\label{line:rec-call}
}

$G_{\LG}\gets$ the graph $G$ with each 
$S_{v}^i$ contracted to $y_{v}$, 
for each $v\in R_{\SM}^i$.\;
$U_{\LG}\gets U - \cup_{v\in R_{\SM}^i}(S_v^i\cap U)$.\;
$(T_{\LG}, {\algoF}_{\LG}, C_{\LG})\gets$ \ApproxElemConnGHTree{$G_{\LG}, U_{\LG}, F\cup\{y_v\mid v\in R_{\SM}^i\}, \epsilon, \lambda$}.\;\label{line:LG-rec-call}

Construct $T$ from $T_{\LG}$ and $\{T_v\}_{v\in R_{\SM}^i}$ as follows. 
Add edges $({\algoF}_{v}(x_{v}),{\algoF}_{\LG}(y_{v}))$ with weight $\abs{\partial S_{v}^i}$.
The embedding ${\algoF}$ inherits values from 
${\algoF}_{\LG}$ or $\{{\algoF}_{v}\}$. 
$C$ inherits the cuts 
of $C_{\LG}$ and $\{C_v\}_{v\in R_{\SM}^i}$ 
on existing edges, and for the new edges,
$C(({\algoF}_{v}(x_{v}),{\algoF}_{\LG}(y_{v}))) = \partial S_v^i$.\;\label{line:combine-GH-trees}

\Return $(T, {\algoF}, C)$.\;
\end{algorithm}

\begin{algorithm}
\caption{\textsc{CutThresholdStep}$(G=(V, E), U, F_{\operatorname{in}}, W)$}
\label{algo:cut_threshold_step}
\DontPrintSemicolon
Set $s\gets $ a uniformly random vertex in $U$.\;
$R^0\gets U$.\;
\For{$j$ from $0$ to $\lowerint{\log\abs{U}}$}{
    Call algorithm of \cref{lem:isolating_cut_lemma_with_forbidden_terminals} to compute sets $\{S^j_v: v\in R^j\}$, with $I=R^j$ and $F=(U - R^j)\cup F_{\operatorname{in}}$.\;
    Let $R_{\SM}^j\gets \{v\in R^j-\{s\}: \abs{S_v^j\cap U}\leq \abs{U}/2 \text{ and } \abs{\partial S_v^j} \leq W\}$.\;
    $R^{j+1}\gets $ sample each vertex of $R^j$ with probability $\frac{1}{2}$, but $s$ with probability 1.\;
}
\Return $s$ and $\{R^j, R_{\SM}^j$, \{$S^j_v\}_{v \in R^j_{\SM}}\}$, for each $j$.\;
\end{algorithm}

\subsection{Correctness}\label{sect:GH-correctness}

Let $F_v=F\cup\{x_v\}$
and 
$F_{\LG}=F\cup\{y_v \mid v\in R_{\SM}^i\}$ 
be the $F$-sets in the recursive calls 
at Lines~\ref{line:rec-call}
and~\ref{line:LG-rec-call}, respectively.

\begin{lemma}
\label{lem:elem_conn_preserved}
The element connectivity of any two 
$p, q\in U_{\LG}$ is preserved exactly in $G_{\LG}$, that is,
\[\kappa'_{G_{\LG}, U_{\LG}\cup F_{\LG}}(p, q) = \kappa'_{G, U\cup F}(p, q).\]
\end{lemma}
\begin{proof}
We first show that $\kappa'_{G_{\LG}, U_{\LG}\cup F_{\LG}}(p, q) \geq \kappa'_{G, U\cup F}(p, q)$.
Suppose $C$ is a minimum element cut disconnecting $p$ and $q$ in $G_{\LG}$ with terminal set $U_{\LG}\cup F_{\LG}$. Then $C$ does not contain any $y_v\in F_{\LG}$, so $C$ is still an element cut for $p$ and $q$ in $G$, and therefore $\kappa'_{G_{\LG}, U_{\LG}\cup F_{\LG}}(p, q) =\abs{C} \geq \kappa'_{G, U\cup F}(p, q)$.

It remains to show that $\kappa'_{G_{\LG}, U_{\LG}\cup F_{\LG}}(p, q) \leq \kappa'_{G, U\cup F}(p, q)$. 
Suppose $C$ is a minimum element cut disconnecting 
$p$ and $q$ in $G$ with terminal set $U\cup F$, 
and let $A,B$ be the sides containing $p,q$, respectively. 
We want to show that $C$ exists in $G_{\LG}$, that is,
it does not intersect with any of the $S_v^i$'s for $v\in R_{\SM}^i$.\footnote{Note that $i$ is one fixed value in $\{0, 1, \cdots, \lowerint{\log\abs{U}}\}$ in each single call of \cref{algo:construct_gh_tree}, while $v$ ranges over $R_{\SM}^i$.}
Fix any $v\in R_{\SM}^i$, and let $D_v$ be the side
of $C$ containing $v$.  Note that $C$ is disjoint
from $U\cup F$, and therefore cannot contain $v\in U$.

Without loss of generality we assume $v$ 
and $q$ are in different sides. 
All the element cuts introduced below 
are w.r.t.~$G$ and $U\cup F$. 
Define $H = D_v\cup A$. 
\cref{lem:partial-of-side-and-its-cut} implies that
$\partial D_v \subseteq C = \partial A$, and since vertices in $D_v$ and $A$ are not adjacent, we also have
$\partial H = \partial D_v\cup \partial A = C$.
Furthermore, we have that $\partial H$ is 
a \emph{minimum} element cut 
between $\{v, p\}$ and $\{q\}$.
Suppose $C'$ is the minimum element cut between $\{v, p\}$ and $\{q\}$.
Since $C'$ disconnects $p$ and $q$, $\abs{C'}\geq \abs{C}$, 
and since $C$ disconnects 
$\{q\}$ from $\{v, p\}$, 
$\abs{C'}\leq \abs{C}$. 
Combining them we conclude 
that $\abs{C'} = \abs{C}$.

Since $q\in U_{\LG}$, we know that
$q\notin S_v^i$.  
We also know $q\notin \partial H = C$, $q\notin \partial S_v^i$, and
$q\notin H\cup S_{v}^i$. 
Because $\{v, p\}\subseteq H\cup S_v^i$, we have that $\partial (H\cup S_v^i)$ is 
\emph{some} element cut between $\{v, p\}$ 
and $\{q\}$, 
so $\abs{\partial H} \leq \abs{\partial (H\cup S_v^i)}$. 
By construction, 
$\partial S_v^i$ is a minimum element cut between $\{v\}$ and $R^i-\{v\}$. 
Since $\partial (S_v^i\cap H)$ is 
also an element cut between $\{v\}$ and $R^i- \{v\}$, 
we have 
$\abs{\partial S_v^i} \leq \abs{\partial (H\cap S_v^i)}$. 
However, 
by \cref{lem:submodularity-of-partial} we have
\[
\abs{\partial H} +\abs{\partial S_v^i}\geq \abs{\partial (H\cap S_v^i)} + \abs{\partial (H\cup S_v^i)}.
\]
Thus, this inequality holds with equality, 
and in particular 
$\abs{\partial (H\cap S_v^i)}=\abs{\partial S_v^i}$. 
However, by \cref{lem:isolating_cut_lemma_with_forbidden_terminals},
$S_v^i$ also minimizes $|S_v^i|$ 
among all minimum isolating cuts, 
so $H\cap S_v^i= S_v^i$, 
that is, $S_v^i\subseteq H$.

If $D_v=A$ then $S_v^i\subseteq D_v$.  
If $D_v\neq A$, then in $G-C$, 
$D_v$ is not connected to $A$, while all of $S_v^i$ is connected to $v$
(by \cref{lem:isolating_cut_lemma_with_forbidden_terminals}), so $S_v^i\cap A=\emptyset$ and $S_v^i\subseteq D_v$.
In either case, $S_v^i\subseteq D_v$, 
so when $S_v^i$ is contracted into $x_v$, 
the cut $C\subseteq V - D_v$ is not affected. 
Thus, $C$ is still an element cut between $p$ and $q$ 
in $G_{\LG}$ with terminal set $U_{\LG}\cup F_{\LG}$,
and 
$\kappa'_{G_{\LG}, U_{\LG}\cup F_{\LG}}(p, q) \leq\abs{C} = \kappa'_{G, U\cup F}(p, q)$.
\end{proof}

\begin{lemma}
\label{lem: elem_conn_preserved_for_small_parts}
For any two vertices $p, q\in U_{v}$ where $v\in R_{\SM}^i$, we have
\[\kappa'_{G, U\cup F}(p, q) \leq \kappa'_{G_v, U_v\cup F_v}(p, q) \leq (1+\epsilon)\kappa'_{G, U\cup F}(p, q).\]
\end{lemma}

\begin{proof}
As in the proof of \cref{lem:elem_conn_preserved},
all cuts in $G_v$ w.r.t.~$U_v\cup F_v$ are also valid
cuts in $G$ w.r.t.~$U\cup F$, hence
$\kappa'_{G, U\cup F}(p, q) \leq \kappa'_{G_v, U_v\cup F_v}(p, q)$.

It remains to prove that
$\kappa'_{G_v, U_v\cup F_v}(p, q)\leq (1+\epsilon)\kappa'_{G, U\cup F}(p, q)$.
Suppose $C$ is a minimum element cut disconnecting 
$p$ and $q$ in $G$ with terminal set $U\cup F$ and let $A, B$ be the sides of $p, q$.
Without loss of generality, suppose $A$ 
does not contain $s$. Because $s\in U$, $C$ does not contain $s$.
All element cuts defined below are w.r.t. $G$ and $U\cup F$.

Because $p\in A\subseteq S_v^i\cup A$, $s\notin A$, $s\notin \partial A = C$, and by \cref{lem:isolating_cut_lemma_with_forbidden_terminals}, $s\notin S_v^i$ and $s\notin \partial S_v^i$, we have that $\partial (S_v^i\cup A)$ is \emph{some}
element cut disconnecting $p$ and $s$. Therefore, $\abs{\partial (S_v^i\cup A)} \geq \lambda$. 
Furthermore, by the definition of $S_v^i$, $\abs{\partial S_v^i}\leq (1+\epsilon)\lambda$. Therefore, by \cref{lem:submodularity-of-partial},
\begin{align}
    (1+\epsilon) \lambda + \abs{\partial A}
    \geq \abs{\partial S_v^i} + \abs{\partial A}
    &\geq \abs{\partial (S_v^i\cup A)} + \abs{\partial (S_v^i\cap A)}
    \geq \lambda + \abs{\partial (S_v^i\cap A)}.\label{eqn:approx-iso-cut-1}
\intertext{Since $S_v^i\cap A$ contains $p$ but not $q$, $\partial(S_v^i\cap A)$ is an element cut disconnecting $p$ and $q$, and since it is contained in $G_v$, it is still an element cut in $G_v$ with terminal set $U_v\cup F_v$, so}
    \abs{\partial (S_v^i\cap A)} &\geq \kappa'_{G_v, U_v\cup F_v}(p, q).\label{eqn:approx-iso-cut-2}
\intertext{By definition
of $A,B,$ and $\lambda$,}
\abs{\partial A} &= \kappa'_{G, U\cup F}(p, q)\geq \lambda.\label{eqn:approx-iso-cut-3}
\intertext{Therefore, by (\ref{eqn:approx-iso-cut-1})--(\ref{eqn:approx-iso-cut-3}),}
\kappa'_{G_v, U_v\cup F_v}(p, q) &\leq \epsilon \lambda + \abs{\partial A} \leq (1+\epsilon)\kappa'_{G, U\cup F}(p, q).\nonumber
\end{align}
\end{proof}

\ignore{
\begin{lemma}\label{lem:f_valid}
The assignment ${\algoF}(v)=\bot$ occurs 
if and only if $v$ appears in $C(e)$ for some edge. 
Therefore, ${\algoF}$ value of vertices in $U\cup F$ never equals $\bot$.
\end{lemma}

\begin{proof}
From the construction of $G_{\LG}$ and $G_v$, it can be seen that only the vertices in $\partial S_v^i=C(({\algoF}_v(x_v),{\algoF}_{\LG}(y_v)))$ 
are defined twice (or more), 
so their ${\algoF}$-value are set to $\bot$. 
These are all the vertices such that ${\algoF}(v)=\bot$.
\end{proof}

\begin{remark}
From \cref{lem:elem_conn_preserved}, \cref{lem: elem_conn_preserved_for_small_parts} and \cref{lem:f_valid}, in line 17 of \cref{algo:construct_gh_tree}, the ${\algoF}$-value of $x_v$, $y_v$ are not $\bot$, so linking the sub-trees would be successful. And $\algoF(v)$ for $v\in U$ equals to $\bot$.
\end{remark}
}

\subsection{Running Time Analysis}\label{subsubsec:run_time_analysis}

The running time analysis in this section closely follows Li and Panigrahi~\cite{Li2021Approximate}.
We continue to use all notation defined in \cref{algo:construct_gh_tree} and \cref{algo:cut_threshold_step}
such as $R^i, R_{\SM}^i, \{S_v^i\},G_{\LG},U_{\LG},F_{\LG}$.

\begin{lemma}
\label{lem:pair_below_k_decrease}
Define $P \subset U^2$ with respect to a threshold $W$ to be
\begin{align*}
P  
&=\{(u, v): \kappa'_{G, U\cup F}(u, v)\leq W, \text{and $|A\cap U|\leq |U|/2$ where $A$ is}\\
& \hspace{2cm}
\text{the side of the minimum $u$-$v$ element cut containing $u$}\}.
\intertext{Similarly define $P_{\LG}$ w.r.t. 
$G_{\LG}, U_{\LG}, F_{\LG}$ and $W$. Then}
\E&(\abs{P_{\LG}}) \leq \left(1-\Omega\left(\frac{1}{\log^2\abs{U}}\right)\right)\abs{P}.
\end{align*}
\end{lemma}

\begin{proof}
We need to lower bound the size of 
$Q=P - P_{\LG}$.
Define $U_{\SM}^i=\bigcup_{v\in R_{\SM}^i}(S_v^i\cap U)$, $U_{\SM}=\bigcup_{i=0}^{\lowerint{\log\abs{U}}}U_{\SM}^i$, and
$U^*= \{u \mid (u, s)\in P\}$.
By definition we know $U_{\SM}^i\subseteq U^*$.
We will show that $\E(\abs{Q})\geq \Omega(1/\log^2 |U|)\abs{P}$ follows from 
the following three claims.
\begin{itemize}
    \item [(1)]For each $u\in U^*$, there are at least $\abs{U}/2$ vertices $v$ such that $(u, v)\in P$.
    \item [(2)] $\E(\abs{U_{\SM}})\geq \Omega(\abs{U^*}/\log\abs{U})$.
    \item [(3)] For each pair $(u, v)$, 
    $\Pr(u\in U^* \mid (u,v)\in P) \geq 1/2$.
\end{itemize}

For Claim (1), consider the minimum element cut disconnecting $s$ and $u\in U^*$. There are at least $\abs{U}/2$ terminals $v$ 
not in the same side as $u$, with $(u, v)\in P$.\footnote{These $v$ are not required to all be in $U_{\LG}$.} 
Claim (2) is proved in \cref{lem:claim_cut_threshold}.
Claim (3) holds because $s\in U$ is chosen
uniformly at random and there are at least $\abs{U}-\abs{U}/2=\abs{U}/2$ terminals not in the side of $u$. 
When $s$ is such a terminal, $u\in U^*$. 

The algorithm fixes 
$i\in\{0, 1, \cdots, \lowerint{\log\abs{U}}\}$ 
that maximizes $\abs{U_{\SM}^i}$. 
Putting it all together,

\begin{align*}
    \E(\abs{Q}) &\geq \frac{\abs{U}}{2}\E(\abs{U_{\SM}^i})
     \geq \frac{\abs{U}}{2(\log \abs{U} + 1)}\E(\abs{U_{\SM}})
     \geq \frac{\abs{U}}{2(\log \abs{U} + 1)} \Omega\left(\frac{\abs{U^*}}{\log \abs{U}}\right)
     \geq \Omega\left(\frac{1}{\log^2\abs{U}}\right)\abs{P}.
\end{align*}
For the first inequality,
Claim (1) implies that each
$u\in U^*$ is involved
in $\abs{U}/2$ pairs in $P$,
all of which do not appear in 
$P_{\LG}$ whenever $u\in U_{\SM}^i$.
The second inequality 
follows from the choice of $i$.
The third inequality follows
from Claim (2).
The fourth inequality follows from Claim (3) and the following bound on the size of $P$. 
\[
\abs{P}/2\leq \E(\abs{\{(u, v)\in P: u\in U^*\}})\leq \abs{U}\cdot \abs{U^*}.
\] 
Note each $u$ appears in at most 
$\abs{U}$ pairs of $P$.
\end{proof}

To complete the gap in the proof of \cref{lem:pair_below_k_decrease} we must establish the validity of Claim (2).

\begin{lemma}[Claim (2) restated]\label{lem:claim_cut_threshold}
$\E(|U_{\SM}|)=\Omega(\abs{U^*}/\log\abs{U})$.
\end{lemma}

\begin{proof}
Let us condition arbitrarily on the identity of $s$, and let $T$ be the element connectivity Gomory-Hu tree for terminal set $U$ 
rooted at $s$.
For each vertex $v\in U$, let $U_v$ be the set of terminals in the subtree rooted at $v$.
For a terminal $v\in U$, we find the edge $e(v)$ along the path from $s$ to $v$ with minimum weight, and when not unique, the one with maximum depth. Let $r(v)$ be the deeper endpoint of $e(v)$.
By definition, a terminal $v$ appears in $U^*$ if and only if $w(e(v))\leq W$ and $\abs{U_{r(v)}}\leq \abs{U}/2$, i.e., once we condition on $s$, $U^*$ is fixed and no longer random.

We say that a vertex $v\in U^*$ is \emph{active} if $v\in R^{i(v)}$ where $i(v)=\lowerint{\log\abs{U_{r(v)}}}$. In addition, if $U_{r(v)}\cap R^{i(v)}=\{v\}$, then we say that $v$ \emph{hits} all of the vertices in $U_{r(v)}$, including itself. For completeness, we define vertices in $U- U^*$ to be inactive; they do not hit any vertices.
Now we show that

\begin{itemize}
    \item [(a)] each vertex that is hit is in $U_{\SM}$;
    \item [(b)] the total number of pairs $(u, v)$ for which $v\in U^*$ hits $u$ is $\Omega(\abs{U^*})$ in expectation;
    \item [(c)] each vertex $u$ is hit by at most $O(\log \abs{U})$ vertices in $v\in U^*$.
\end{itemize}

To prove statement (a), suppose $v$ hits some vertices. Then by definition, $U_{r(v)}\cap R^{i(v)} = \{v\}$.  The isolating cut for $v$ returned by \cref{lem:isolating_cut_lemma_with_forbidden_terminals} corresponds
to the edge joining $r(v)$ to its parent, 
so all vertices in 
$U_{r(v)}$ are on $v$'s side of the cut.
Furthermore, since $v\in U^*$, $w(e(v))\leq W$,
and $\abs{S_v^{i(v)}\cap U}=\abs{U_{r(v)}}\leq \abs{U}/2$, it follows that $v\in R_{\SM}^{i(v)}$. 
Hence all vertices hit by $v$ (namely $U_{r(v)}$) appear in $U_{\SM}$.

Turning to statement (b), 
the probability that 
$R^{i(v)}\cap U_{r(v)} = \{v\}$
is $(1-2^{-i(v)})^{\abs{U_{r(v)}}-1}2^{-i(v)}=\Theta(1/2^{i(v)})$, and when this event happens, $v$ hits $\abs{U_{r(v)}}=\Omega(2^{i(v)})$ vertices, so the contribution in the expectation is $\Omega(1)$. Since each $v\in U^*$ contributes $\Omega(1)$ in expectation, their sum is $\Omega(\abs{U^*})$.

We now consider statement (c). 
We first show that for any two vertices 
$v, v'\in U^*$ that both hit $u$, it must be that
$i(v)\neq i(v')$. 
Since $u\in U_{r(v)}$ and $u\in U_{r(v')}$, without loss of generality we assume $r(v)\in U_{r(v')}$, so $U_{r(v)}\subseteq U_{r(v')}$. Recall that the $R$-sets are nested, that is, $R^0\supseteq R^1\supseteq\cdots \supseteq R^{\lowerint{\log\abs{U}}}$, so $R^{i(v)}\cap R^{i(v')} = R^{\max\{i(v), i(v')\}}$. 
Then,
\begin{align*}
    \emptyset &= \{v\}\cap \{v'\} 
    = (R^{i(v)}\cap U_{r(v)})\cap (R^{i(v')}\cap U_{r(v')})
    = R^{\max\{i(v), i(v')\}}\cap U_{r(v)},
\end{align*}
and because $R^{i(v)}\cap U_{r(v)}=\{v\}\neq \emptyset$, we infer that 
$\max\{i(v), i(v')\} = i(v') > i(v)$.
Thus, $u$ can be hit by at most 
$\log \abs{U}$ vertices.
To complete the proof, we have
\begin{align*}
\mathbb{E}[U_{\SM}] 
&\geq \mathbb{E}[\abs{\{u: u\text{ is hit}\}}]
\geq \frac{\mathbb{E}[\abs{\{(u, v): v\in U^*, u \text{ is hit by }v\}}]}{\log \abs{U}}
\geq \Omega\left(\frac{\abs{U^*}}{\log\abs{U}}\right).
\end{align*}

The first inequality is statement (a), 
the second inequality follows from statement (c), and the third inequality is from statement (b).
\end{proof}

\subsection{Proof of \cref{thm:construct_k_gh_tree}
and \cref{thm:construct_approx_gh_tree}}\label{subsec:proof_of_thm_51}

Now we are ready to prove 
\cref{thm:construct_k_gh_tree}
and
\cref{thm:construct_approx_gh_tree}.

\begin{proof}[Proof of \cref{thm:construct_approx_gh_tree}]
The recursion of \cref{algo:construct_gh_tree} makes progress in one of two ways.  In the
``non-$G_{\LG}$'' branches $\{G_v\}$, each $G_v$ contains at most half the number of terminals, so the recursive depth along these branches is at most $\log|U|$.  
Suppose we follow the ``$G_{\LG}$'' 
branch $\Theta(\log^2|U|\log n)$ 
times, yielding $G', U', F'$, and $P'$.  
By \cref{lem:pair_below_k_decrease},
with $W=(1+\epsilon)\lambda$,
\[
\E(|P'|) \leq (1-\Omega(1/\log^2 |U|))^{\Theta(\log^2|U|\log n)}|P| = n^{-\Omega(1)},
\]
meaning 
$P'$ is $\emptyset$ 
w.h.p.~and the global
minimum element cut of $G',U'\cup F'$ 
has increased to at least $(1+\epsilon)\lambda$ and 
we can update $\lambda$ accordingly.
\cref{algo:construct_gh_tree} updates $\lambda$ every 
$\Delta = \Theta(\log^3 n)$ steps in the $G_{\LG}$ branch.

This implies
the total depth of the recursion is $O(\epsilon^{-1}\log^3|U|\log n) = O(\epsilon^{-1}\log^4 n)$ w.h.p.
The total size of all graphs on each layer
of recursion is $O(m)$, hence 
by~\cref{lem:isolating_cut_lemma_with_forbidden_terminals},
the total time is 
$O(\epsilon^{-1}\log^4 n\cdot T_{\flow}(m)\log n)$.

Turning to correctness, by \cref{lem:elem_conn_preserved} 
the $G_{\LG}$-branch preserves the 
exact value of $\kappa'_{G, U\cup F}$, and all the non-$G_{\LG}$ branches $\{G_v\}$ introduce a 
$(1+\epsilon)$-factor approximation to
the element connectivity. 
Since the depth of the recursion in the non-$G_{\LG}$ branches is 
at most $\log |U|$, the tree returned
is a $(1+\epsilon)^{\log|U|}$-
approximate Gomory-Hu tree.

By picking $\epsilon = \epsilon_{0} / (2\log \abs{U})$ for $0 < \epsilon_0 < 1$, the approximation ratio is $(1+\epsilon)^{\log|U|} \leq e^{\epsilon \log\abs{U}} = e^{\epsilon_0/2} \leq 1 + \epsilon_0$. Meanwhile, expressed in terms of $\epsilon_0$, the running time is 
$O(\epsilon_0^{-1}\log^6 n\cdot T_{\flow}(m))$.
\end{proof}

\begin{proof}[Proof of \cref{thm:construct_k_gh_tree}]
The proof basically follows from the proof of \cref{thm:construct_approx_gh_tree} by setting $\epsilon_0 = 1/(k+1)$.  Cut values are integers.
Thus, in the output Gomory-Hu tree, all queries that report a cut-value in $[1,k]$ are exact, 
and we are free to contract any edge with weight $k+1$ or larger.
\end{proof}

\section{Conclusion}\label{sect:conclusion}

In this paper we proved that $\Omega(kn)$ 
bits of space are necessary for encoding vertex connectivity information 
up to $k$.  This establishes the near-optimality of several 
previous results.  
For example, Nagamochi-Ibaraki~\cite{NagamochiI92} sparsifiers 
encode all vertex connectivities up to $k$, but their space 
cannot be improved by more than a $\log n$ factor, 
\emph{even if the format of the representation 
is not constrained to be a graph}.  
It also implies that even the 
\emph{average} length of the Izsak-Nutov~\cite{IzsakN12} 
labeling scheme cannot be improved much.
We improved~\cite{IzsakN12} 
to have near-optimal query time $O(\log n)$,
independent of $k$,
and improved 
its construction time to 
nearly max-flow time.

Here we highlight a few open problems.
\begin{itemize}
    \item There is a trivial $\Omega(kn)$ space lower bound for data structures answering $\vcut$ queries.  Our data structure (and Izsak-Nutov~\cite{IzsakN12}) can be augmented to support fast $\vcut$ queries with $O(k^2n\log n)$ space.  Is this necessary?  Note that if it is, the lower bound cannot be purely information-theoretic; it must hinge on the requirement that queries be answered \emph{efficiently}.\footnote{There are two natural ways to answer $\vcut$ queries ``optimally,'' either enumerate their elements in $O(k)$ time, or return a pointer in $O(1)$ time 
    to a pre-stored list of elements.  The latter model was recently advocated by Nutov~\cite{Nutov22} and seems to be easier to characterize from a lower bound perspective.}

    \item A special case of the vertex connectivity oracle problem is 
    answering $\vconn(u,v)$ queries when $k=\kappa(G)$.  In other words, decide whether or not 
    $u,v$ are separated by a \emph{globally} minimum cut.
    Globally minimum vertex cuts have plenty of structure~\cite{PettieY21,cohen1993reinventing}, but prior results do not imply that beating $kn$ bits is possible.  We conjecture that $\tilde{O}(n)$ bits is possible.
    
    \item Is there a $(1+\epsilon)$-approximate vertex connectivity oracle with space $\tilde{O}(n/\epsilon^2)$, matching the lower bound of \cref{cor:approximation-lb}? 
    We conjecture that this is possible.

    \item Is there a max-flow-time algorithm for constructing an \emph{exact} Gomory-Hu tree for element connectivity?
\end{itemize}

\paragraph{Acknowledgment.} 
We would like to thank the three anonymous reviewers, whose comments substantially improved the quality of the presentation.